\documentclass[final,3p,times]{elsarticle}

\usepackage{graphicx,times,psfig}         
\usepackage{amssymb}
\usepackage{amsmath}
\usepackage{amsthm}
\usepackage{float}
\usepackage{epstopdf}
\usepackage{cases}
\usepackage{color}

\usepackage{subcaption}

\newtheorem{definition}{Definition}
\newtheorem{theorem}{\textbf{Theorem}}[section]
\newtheorem{lemma}[theorem]{\textbf{Lemma}}
\newtheorem{assumption}[theorem]{Assumption}

\newtheorem{remark}{Remark}

\usepackage[labelsep=period]{caption}

\biboptions{sort&compress}

\newcommand{\mas}{%
{multi-agent system}}

\newcommand{\ct}{%
{communication topology}}

\journal{}

\begin{document}

\begin{frontmatter}
\title{Distributed consensus of linear MASs with an unknown leader via a predictive extended state observer considering input delay and disturbances}

\author[label0]{Wei Jiang}
\ead{wjiang.lab@gmail.com}
\author[label1,label11,label111]{Zhaoxia Peng
}
 \ead{pengzhaoxia@buaa.edu.cn}
\author[label0]{Ahmed Rahmani}\ead{ahmed.rahmani@centralelille.fr}
\author[label0]{Wei Hu}\ead{huw0906@gmail.com}
\author[label2]{Guoguang Wen}\ead{guoguang.wen@bjtu.edu.cn}
\address[label0]{CRIStAL, UMR CNRS 9189, Ecole Centrale de Lille, Villeneuve d'Ascq, France}
\address[label1]{School of Transportation Science and Engineering, Beihang University, Beijing, 100191,P.R.China}
\address[label11]{Beijing Engineering Center for Clean Energy \& High Efficient Power, Beihang University, Beijing 100191, P.R.China}
\address[label111]{Laboratoire international associ$\acute{e}$, Beihang Universitym, P.R.China}
\address[label2]{Department of Mathematics, Beijing Jiaotong University, Beijing 100044, P.R.China}

%
%
%
\begin{abstract}
The problem of disturbance rejection/attenuation for constant-input delayed linear multi-agent systems (MASs) with the directed communication topology is tackled in this paper, where a classic model reduction technique is introduced to transform the delayed MAS into the delay-free one. First, when the leader has no control input, a novel adaptive predictive extended state observer (ESO) using only relative state information of neighboring agents is designed to achieve disturbance-rejected consensus tracking. The stabilization analysis is presented via the Lyapunov function and sufficient conditions are derived in terms of linear matrix inequalities. Then the result is extended to disturbance-attenuated case where the leader has bounded control input which is only known by a portion of followers. Finally, two numerical examples are presented to illustrate the effectiveness of proposed strategies. The main contribution focuses on the design of adaptive predictive ESO protocols with fully distributed property.
\end{abstract}


\end{frontmatter}

\section{Introduction}
{D}{istributed} cooperative control has gained increased research attention due to its widely potential applications such as unmanned aerial vehicle formation, complex networks synchronization~\cite{wang10pinning}, satellite clusters and so on.
Among different kinds of cooperative control formats, consensus control, which aims at controlling all agents to achieve the same objective, has been investigated tremendously thanks to the impressive framework-building works~\cite{jadbabaie_coordination_2003,olfati-saber_consensus_2004,ren_consensus_2005}. After that, many consensus results have been presented from undirected to directed communication topology concerning different dynamics, for instance, fractional-order~\cite{bai2017distributed}, first-order, double integrator~\cite{wen2016dynamical,jiang2018asynchronous}, second-order, general linear~\cite{cheng2016convergence,wu2017event} and nonlinear dynamics. Particularly, Li~\cite{li2010consensus} presented a unified framework by expanding conventional observers to distributed observers, which is a significant breakthrough to solve the consensus problem and synchronization of complex networks. Then, the fully distributed consensus control, which needs no global information like the minimum eigenvalue of Laplacian matrix of communication topology, was proposed in~\cite{li_distributed_2013-1}. It is worth noting that the fully distributed property in the consensus controller is very meaningful as it is nearly impossible for each agent to know the Laplacian matrix for protocol calculating when the number of agents is very large. Except the work~\cite{olfati-saber_consensus_2004} with first-order and the work~\cite{wen2016dynamical,jiang2018asynchronous} with double integrator dynamics, which dealt with time-delay consensus problem , all the works above do not cover the disturbance or time-delay issues.

The external disturbance is widely existed in industrial control process and thus has been researched for decades. One of the fundamental idea is to design an observer mechanism to estimate those disturbances and then incorporate the designed observer into input controller to compensate the effects of disturbances.
Readers are recommended to survey papers~\cite{madonski2015survey,chen_disturbance-observer-based_2016} about extended state observers (ESOs) for disturbance attenuating and rejecting. In this paper, we embrace the same idea to design ESOs. The latest work concerning the fully distributed consensus tracking with disturbance rejection, is in~\cite{sun2016distributed} where the leader is treated without control input which means that the leader's dynamics is known once its initial state is known. In our work, we deal further with an unknown leader which can be interpreted with an existing unknown control input. What is more, the work~\cite{sun2016distributed} does not consider the time-delay issue which exists commonly in networked control systems and can deteriorate the system stability heavily. Our work also covers the control input delay problem which can be regarded as another improvement.

In terms of the time-delay system which is the longterm interest in control community because of its wide existence in reality, readers can refer to survey papers~\cite{richard2003time,gu2003survey} for an overview about this topic. In the MAS, time-delay mainly occurs inside the input, state, output or communication. This work deals with the constant input delay. Here two different approaches exist: memoryless (memory free) and memory controllers~\cite{lechappe2015new}. one of the typical memoryless approaches is the truncated prediction feedback (TPF) approach which is originally proposed by Lin~\cite{lin2007asymptotic} and further developed in~\cite{zhou2014consensus,yoon2013truncated} with the advantage that the requirement for integral computation can be removed.
However, it requires the open-loop system satisfying the following conditions: polynomially unstable (all eigenvalues of system matrix on the imaginary axis) with any large input delay~\cite{zhou2014consensus} or exponentially unstable with small enough input delay~\cite{yoon2013truncated}.
The disadvantage is that the control input is not utilized efficiently and sufficiently as the value of input is usually quite small due to the incorporation of low gain feedback technique. On the other hand, memory controller design which incorporates the computation of variables' historical values for prediction regains researchers' attention recently. As said in~\cite{mirkin2003every}, \textit{state prediction is a fundamental concept for delay systems, much like state observation is for systems with incomplete state measurements}. For linear systems, two popular approaches are the \textit{Smith predictor}~\cite{smith1957close} in frequency-domain and the \textit{reduction} technique in time-domain. With earlier impressive work in~\cite{mayne1968control,manitius1979finite,kwon1980feedback} and then systematically generalized in Artstein~\cite{artstein1982linear}, the reduction technique is frequently utilized in the control input delay system due to the fact that the system can be transformed into a delay-free one for the convenience of controller designing. This technique is thus used in this paper. Recently, the leaderless consensus~\cite{wang2016h} and leader-follower consensus~\cite{wang_consensus_2017-1,wangchunyanCyber} considering constant input delay and disturbances were investigated. The drawback is that the latter results assume the leader without control input (not an unknown leader). In addition, all the above three latest results are not fully distributed since the parameter designing inside control protocols is related to the Laplacian matrix of communication topology which is a piece of global information. Another newest work~\cite{ponomarev2017discrete} was the leaderless consensus using discrete-time predictor feedback technique. However, the disturbance which may deteriorate the controlling stability is not considered and the protocol is not fully distributed as well.

Based on the above analysis, the difficulty in this paper arises from the fully distributed protocol design with the unknown leader for consensus tracking control considering input delay and disturbances, where the corresponding solution is also the main contribution. To do this, a novel adaptive predictive ESO is first proposed to deal with the case of leader without control input. A linear matrix inequality (LMI) is obtained to prove the effectiveness of proposed controller. Then a modified protocol is introduced to tackle the unknown leader-follower consensus. 
To author's best knowledge, this is the first time proposing the fully distributed controller to tackle consensus tracking with an unknown leader under the directed communication topology considering input delay and disturbances.

\section{Preliminaries and model formulation}\label{section2}

\subsection{Mathematical preliminaries}

The connections between agents can be represented by a weighted graph $ \mathcal{G} = (\mathcal{V,E,A})$, where $\mathcal{V}$ and $\mathcal{E}$ denote the nodes and edges, respectively. $\mathcal{A} = [ a_{ij} ] \in \mathbb{R} ^{N \times N}$ denotes the adjacency matrix where $a_{ij}=1 $ if there exists a path from agent $ j $ to agent $ i $, and $a_{ij} = 0$ otherwise. An edge $\left(i,j\right) \in \mathcal{E}$ in graph $ \mathcal{G}$ means that agent $j$ can receive information from agent $i$ but not necessarily conversely.
The Laplacian matrix $\mathcal{L} = [ l_{ij} ] \in \mathbb{R} ^{N \times N}$ is normally defined as $l_{ii}= \sum_{j \neq i} a_{ij} $ and $l_{ij}= -a_{ij} $ when $ i \neq j$. 
A directed path from node $i$ to $j$ is a sequence of edges $\left(i,i_{1}\right),\left(i_{1},i_{2}\right), \ldots,\left(i_{k},j\right)$ with different nodes $i_{s}, s=1,2,\ldots, k$. 
A directed graph contains a directed spanning tree if there is a node from which a directed path exists to each other node.
More graph theories can be found in~\cite{godsil_algebraic_2001}.

The symbol $ \textbf{1} $ denotes a column vector with all entries being 1.
Matrix dimensions are supposed to be compatible if not explicitly stated.
The symbol $\otimes$ represents the Kronecker product and $ diag\{a_1, \ldots, a_{n}\} $ denotes a diagonal matrix with the diagonal entries being $  a_1, \ldots, a_{n}$.
The matrix $ A=[a_{ij}] \in \mathbb{R}^{N\times N} $ is called a nonsingular $ M $-matrix if $ a_{ij} \le 0, \forall i\neq j $, and all eigenvalues of $ A $ have positive real parts. Here, $ \lambda_{min}(A) $ and $ \lambda_{max} (A)$ represent the minimal and maximal eigenvalues of $ A $, respectively.
Suppose the eigenvalues of $ S\in \mathbb{R}^{n\times n} $ and $ T\in \mathbb{R}^{m\times m} $ are $ \lambda_{1},\ldots,\lambda_{n} $ and $ \mu_{1},\ldots, \mu_{m} $, respectively, then the eigenvalues of $ S\otimes T $ are $ \lambda_{i}\mu_{j}, i=1,\ldots,n, j=1,\ldots, m $.
If $ S \in \mathbb{R}^{n\times n}$ and $ T\in \mathbb{R}^{n\times n} $ are two symmetric positive definite matrices, then $ \lambda_{max}(ST)\le \lambda_{max}(S)\lambda_{max}(T) $.
$ det(A) $ is the determinant of a square matrix $ A $, and $ tr(A) $ is defined to be the sum of the elements on the main diagonal of $ A $. For a vector $ x $, denote $ \|x\| $ as its 2-norm. 
For any integer $ a \le b $, denote $ \textbf{I}[a,b]=\{a, a+1,\ldots, b\} $.
\begin{lemma}[\cite{qu_cooperative_2009}]\label{lemma_M_matrix}
	For a nonsingular $ M $-matrix $ A $, there exists a positive diagonal matrix $ G = diag(g_{1}, \ldots, g_{N}) > 0 $ such that $ GA+A^{T}G > 0 $.
\end{lemma}

\begin{lemma}[\cite{bernstein_matrix_2009}]\label{lemma_pq}
	If $ a, b $ are nonnegative real numbers and $ p, q $ are positive real numbers satisfying $ \frac{1}{p}+\frac{1}{q}=1 $, then $ ab \leq \frac{a^{p}}{p}+\frac{b^{q}}{q} $.
\end{lemma}

\begin{lemma}[\cite{corless1981continuous}]\label{lemma_kappa_function}
	For a system $ \dot{x}=f(x,t) $ where $ f(\cdot) $ is locally Lipschitz in $ x $ and piecewise continuous in $ t $,  suppose that there exists a continuously differentiable function $ V(x,t) \ge 0 $ satisfying
	\begin{equation*}
	\begin{aligned}
	\mathcal{K}_{1}(\|x\|) \le V(x,t) \le & \mathcal{K}_{1}(\|x\|)\\
	\dot{V}(x,t) \le& -\mathcal{K}_{3}(\|x\|) + \Xi
	\end{aligned}
	\end{equation*}
	where $ \Xi >0 $ is a constant, $ \mathcal{K}_{1}, \mathcal{K}_{2} $ belong to class $ \mathcal{K}_{\infty} $ functions, and $ \mathcal{K}_{3} $ belongs to class $ \mathcal{K} $ function. Then, the solution $ x(t) $ of $ \dot{x}=f(x,t) $ is uniformly ultimately bounded.
\end{lemma}

\begin{lemma}[\cite{khalil1996noninear}]\label{lemma_dotV_V}
	If a real function $ V(t) $ satisfies $ \dot{V}(t) \le -aV(t)+b $, where $ a, b $ are positive constants, then
	$$
	V(t)\le (V(0)-\frac{b}{a})e^{-at}+\frac{b}{a}.
	$$ 
\end{lemma}

\subsection{Model formulation}
In this subsection, a group of $N+1$ agents with identical linear dynamics is described as
\begin{equation}\label{eq:dynamics}
\dot x_{i}(t)  =  Ax_{i}(t)+Bu_{i}(t-\tau)+Ew_{i}(t),\quad i \in \textbf{I}[0,N]
\end{equation}
where $x_{i}(t)=[x_{i1}(t), \ldots, x_{in}(t)]^{T} \in \mathbb{R} ^{n}$ and $u_{i}\left(t\right) \in \mathbb{R} ^{p}$ are the state, control input of the $ i $-th agent, respectively. $ A \in \mathbb{R} ^{n \times n}, B \in \mathbb{R} ^{n \times p} $ and $ E \in \mathbb{R} ^{n \times s} $ are constant matrices. $ \tau $ is the system's control input delay. $ w_{i}(t) \in \mathbb{R}^{s} $ is the corresponding external disturbance which is generated by the following exosystem
\begin{equation}\label{w_disturbance}
\dot{w}_{i}(t) = S w_{i}(t),\quad i \in \textbf{I}[0,N]
\end{equation}
with $ S \in \mathbb{R} ^{s \times s} $ being a known constant matrix.

\begin{assumption}\label{assumption_tau}
	$ \tau >0 $ is constant and known.
\end{assumption}
\begin{remark}
	Here the delay variable $ \tau $ is taken as fixed and identical for the written convenience, but it can be extended to time-varying input delay as long as we know the upper bound of the delay. The work on unknwon time-varying input delay by the observer estimating technique is undergoing.
\end{remark}
\begin{assumption}\label{assumptionABC}
	$ (A,B)  $ is controllable. 
\end{assumption}
Without loss of generality, suppose that agents in \eqref{eq:dynamics} indexed by $ 1, \ldots, N $ are the followers denoted as $ \mathbb{F} \triangleq \left\lbrace 1,\ldots,N\right\rbrace  $ and the agent indexed by $ 0 $ is the leader which receives no information from the followers. Note that the leader's state information is only available to a subset of followers.
The leader is regarded without the control input in Subsection~\ref{subsec_consensus_without_leader_input}, i.e., $ u_{0}(t)=0 $, which is a common assumption in the existing works on distributed cooperative control of linear MASs \cite{wang_consensus_2017-1,wen_containment_2016,ding2015consensus,wangchunyanCyber}. 


However, as we know, where the whole \mas{} moves is decided by the leader and that is why the leader exists. Then where will the leader move? The answer is that a desired dynamic trajectory command is given to the leader to ask the leader to finish the desired trajectory tracking or that the leader moves wherever it can, which requires the leader's control input to be nonzero.
But if the leader always has no control input, it means the leader is a virtual one and hence its tracking ability has severe limitations because of the equation $ \dot x_{0}(t)  =  Ax_{0}(t)+Ew_{0}(t) $ as the system matrices $ A $ and $ E $ are unchangeable with $ w_{0}(t) $ being the leader's external disturbance.
In real applications, the leader needs to regulate the final consensus trajectory. So its control input $ u_{0}(t) $ will not be affected by followers. 
In Subsection~\ref{subsec_consensus_with_leader_input}, we deal with the disturbance-attenuating consensus control in a fully distributed fashion considering the leader's input satisfying the following assumption, which is more difficult than the case of $ u_{0}(t)=0 $.

\begin{assumption}\label{leader input}
	The leader's control input satisfies that $ \|u_{0}(t)\| \le \epsilon $, where $ \epsilon $ is a positive constant.
\end{assumption}
\begin{assumption}\label{assumptiondirected}
	The graph $ \mathcal{G} $ contains a directed spanning tree where the leader acts as the root node.
\end{assumption}

Then the Laplacian matrix of $ \mathcal{G} $ can be partitioned as $ \mathcal{L} =  
\begin{bmatrix}
0 & 0_{1 \times N} \\
\mathcal{L}_{2} & \mathcal{L}_{1}  
\end{bmatrix} $, where $ \mathcal{L}_{1} \in \mathbb{R}^{N \times N} ,  \mathcal{L}_{2} \in \mathbb{R}^{N \times 1}$. Under Assumption~\ref{assumptiondirected}, all the eigenvalues of $  \mathcal{L}_{1} $ have positive real parts \cite{cao_distributed_2012}. It is also easy to confirm that $ \mathcal{L}_{1} $ is a nonsingular $ M $-matrix~\cite{qu_cooperative_2009}.

\begin{assumption}\label{assum_SD}
	There exists a matrix $ F\in \mathbb{R} ^{p \times s} $ such that $ E=BF $, meaning that the disturbance is matched. The eigenvalues of $ S $ are distinct and on the imaginary axis. 
	$ (S,E) $ is observable.
\end{assumption}

The assumption of eigenvalues of $ S $ assures the external disturbance $ w_{i}(t), i\in \textbf{I}[0,N] $ to be the non-vanishing harmonic disturbance including constants and sinusoidal functions,
which is commonly used for output regulation and disturbance rejection. 
In addition, the matched disturbances could be relaxed and be transformed to unmatched ones in some circumstances~\cite{isidori1995nonlinear}.
The detailed explanation of Assumption~\ref{assum_SD} can be refered to the Remark 1 in~\cite{ding2015consensus}.

If $ w_{i}(t), i\in \textbf{I}[0,N] $ is known, the disturbance rejection is quite straightforward by adding the term $ -F(w_{i}(t)-w_{0}(t)) $ in the control input $ u_{i}(t) $. The key issue here is to design fully distributed observers to estimate those unknown disturbances under the directed communication topology $ \mathcal{G} $ satisfying Assumption~\ref{assumptiondirected}. The disturbance state $ w_{i}(t) $ is expected to be observable from the system state measurement $ x_{i}(t), i\in \textbf{I}[0,N] $. For this purpose, inspired by~\cite{ding2015consensus}, we propose the following lemma.
\begin{lemma}\label{lemma_AeAtau}
	If $ (S,E) $ is observable, then the pair $ (A_{T},T) $ is observable, with $ A_{T}=\begin{bmatrix}
	A & e^{A\tau}E \\
	0 & S 
	\end{bmatrix} $ and $ T=\begin{bmatrix}
	I & 0  
	\end{bmatrix} $.
\end{lemma}
\begin{proof}
	Let us prove the result by seeking a contradiction. Assume that $ (A_{T},T) $ is not observable, for any eigenvalue of $ A_{T} $, i.e., $ \lambda_{i} $, the matrix$$
	\begin{bmatrix}
	\lambda_{i}I-A & -e^{A\tau}E \\
	0 & \lambda_{i}I-S \\
	I & 0
	\end{bmatrix}
	$$
	is rank deficient, i.e., there exists a nonzero vector $ \eta=[\eta_{1}^{T}, \eta_{2}^{T}]^{T} \in \mathbb{R} ^{(n+s)} $ such that
	$$
	\begin{bmatrix}
	\lambda_{i}I-A & -e^{A\tau}E \\
	0 & \lambda_{i}I-S \\
	I & 0
	\end{bmatrix}\begin{bmatrix}
	\eta_{1} \\
	\eta_{2}
	\end{bmatrix}=0.
	$$
	This implies that
	\begin{equation}\label{AT}
	\eta_{1}=0, \, \begin{bmatrix}
	-e^{A\tau}E \\
	\lambda_{i}I-S \\
	\end{bmatrix}\eta_{2}=0.
	\end{equation}
	Since $ \eta_{1}=0 $, we get $ \eta_{2} \neq 0 $. 
	
	It is known that $ det(e^{A\tau}) = e^{tr(A\tau)} >0 $, which means $ e^{A\tau} $ is invertible, i.e., $ rank(e^{A\tau})=n $. From $ -e^{A\tau}E\eta_{2}=0 $ in~\eqref{AT} we have $ \begin{bmatrix}
	-E \\
	\lambda_{i}I-S \\
	\end{bmatrix}\eta_{2}=0 $, which implies, together with $ \eta_{2} \neq 0 $, that $ (S,E) $ is not observable. This is a contradiction, meaning that $ (A_{T},T) $ must be observable.
\end{proof}

Since $ (A_{T},T) $ is observable, there exists a positive definite matrix $ P $ that satisfies the following LMI
\begin{equation}\label{eq:ARE}
PA_{T}+A_{T}^{T}P-2T^{T}T < 0.
\end{equation}

\section{Main results}\label{mainresults}

This section mainly focuses on how to design fully distributed adaptive protocols to address consensus tracking problems considering input delay and disturbances with the directed communication topology. 
Subsection~\ref{subsec_consensus_without_leader_input} solves the consensus tracking problem with the leader of no control input based only on relative state measurements. 
After that, the extended case of the leader with bounded input is studied in Subsection~\ref{subsec_consensus_with_leader_input}.


The control goal here is to design fully distributed protocols to make followers track the leader based only on relative states under the directed communication topology $ \mathcal{G} $.
To do this, define the consensus tracking error for follower $ i $ as $ \tilde x_{i}(t)=x_{i}(t)-x_{0}(t) $. The objective here is to prove the convergence of $ \tilde{x}_{i}(t)$ for any initial state $ x_{0}(0) $ and $x_{i}(0) , i \in \mathbb{F} $.

\subsection{Consensus tracking control without $ u_{0}(t) $}\label{subsec_consensus_without_leader_input}

The dynamics of $ \tilde{x}_{i}(t) $ is
\begin{equation}\label{x_tilde}
\dot{\tilde{x}}_{i}(t)  = A\tilde{x}_{i}(t)+Bu_{i}(t-\tau)+E\bar{w}_{i}(t), i \in \mathbb{F}
\end{equation}
where $ \bar{w}_{i}(t)=w_{i}(t)-w_{0}(t) $.
Here, we concern about the disturbance rejection $ \bar{w}_{i}(t) $ and control input delay $ u_{i}(t-\tau) $.

Firstly, if there is no input delay and suppose the disturbance $ w_{i}, i \in \textbf{I}[0,N] $ is known, the method of disturbance rejection is quite easy by adding a term $ -F\bar{w}_{i}(t) $ in $ u_{i}, i \in \mathbb{F} $. So the key technique is to estimate $ \bar{w}_{i}(t) $ by designing a fully distributed observer $ \hat{w}_{i}(t) $. This is one of main contributions in this paper and will be explained in detail later.

Then, in terms of input delay $ u_{i}(t-\tau), i \in \mathbb{F} $, inspired by the reduction technique~\cite{artstein1982linear,lechappe2015new} which can be utilized and modified to transform the system~\eqref{eq:dynamics} with a delayed input into a delay-free system, the variable transformation for each follower $ i $ is designed as follows
\begin{equation}\label{z_transformed_variable}
\tilde{Z}_{i}(t)
= e^{A\tau}\tilde{x}_{i}(t)
+ \int_{t-\tau}^{t} e^{A(t -s)}[Bu_{i}(s) + Ee^{S\tau}\hat{w}_{i}(s)]ds.
\end{equation}
\begin{remark}
	Here, the link between the consensus tracking error $ \tilde{x}_{i}(t) $ and transformed variable $ \tilde{Z}_{i}(t) $ is established, which is one of the main difficulties in this paper.
\end{remark}

Let us define an augmented state $ Z_{i}(t)=[\tilde{Z}_{i}(t)^{T},\bar{w}_{i}(t)^{T}]^{T} $ and apply the transformation~\eqref{z_transformed_variable} on system~\eqref{x_tilde}, then
\begin{equation}\label{eso}
\begin{aligned}
\dot{Z}_{i}(t)=&\underbrace{\begin{bmatrix}
	A & e^{A\tau}E \\
	0 & S 
	\end{bmatrix}}_{A_{T}}Z_{i}(t)+\underbrace{\begin{bmatrix}
	B \\
	0 
	\end{bmatrix}}_{\bar{B}}u_{i}(t) +\begin{bmatrix}
Ee^{S\tau} \\
0 
\end{bmatrix}\hat{w}_{i}(t)- \begin{bmatrix}
e^{A\tau}Ee^{S\tau} \\
0 
\end{bmatrix}\tilde{w}_{i}(t-\tau)
\end{aligned}
\end{equation}
where $ A_{T}\in  \mathbb{R}^{(n+s) \times (n+s)}, \bar{B} \in  \mathbb{R}^{(n+s) \times p} $.

The idea is to design the fully distributed adaptive ESO as $ \bar{Z}_{i}(t)=[v_{i}(t)^{T}, \hat{w}_{i}(t)^{T}]^{T}, i\in \mathbb{F} $ to estimate the extended state $ Z_{i}(t)=[\tilde{Z}_{i}(t)^{T},\bar{w}_{i}(t)^{T}]^{T} $, which will be elaborated in detail in the following.
According to~\eqref{eso}, the control input for each follower $ i $ could be designed as
\begin{equation}\label{input}
u_{i}(t)=(\underbrace{\begin{bmatrix}
	K_{1} & 0 
	\end{bmatrix}}_{\bar{K}_{1}}
-\underbrace{\begin{bmatrix}
	0 & Fe^{S\tau} 
	\end{bmatrix}}_{\bar{F}}
)\bar{Z}_{i}(t), i \in \mathbb{F}
\end{equation}
such that
\begin{equation}\label{dot_Z_2}
\dot{\tilde{Z}}_{i}(t)=(A+BK_{1})\tilde{Z}_{i}(t)+BK_{1}\tilde{v}_{i}(t)-e^{A\tau}Ee^{S\tau}\tilde{w}_{i}(t-\tau)
\end{equation}
where $ \tilde{v}_{i}(t)=v_{i}(t)-\tilde{Z}_{i}(t) $ and $ \tilde{w}_{i}(t)=\hat{w}_{i}(t)-\bar{w}_{i}(t) $ are observer estimating errors,
and $ K_{1}\in \mathbb{R}^{p \times n} $ is a constant matrix to be designed later. 

On the other hand, where is the link among consensus tracking error $ \tilde{x}_{i}(t) $, transformed variable $ \tilde{Z}_{i}(t) $, the ESO $ \bar{Z}_{i}(t) $ and the designed control input $ u_{i}(t) $? The answer is to substitute the designed control input~\eqref{input} into the transformed delay-free system~\eqref{z_transformed_variable}, then
\begin{equation}\label{most_important_equation}
\tilde{Z}_{i}(t)=e^{A\tau}\tilde{x}_{i}(t)   + \int_{t-\tau}^{t} e^{A(t -s)}BK_{1}v_{i}(s)ds.
\end{equation}
It is known that $ det(e^{A\tau}) = e^{tr(A\tau)} >0 $, which means $ e^{A\tau} $ is invertible, i.e., $ rank(e^{A\tau})=n $.
So the objective here changes to design the ESO $ \bar{Z}_{i}(t) $ such that $ \lim_{t \to \infty} v_{i}(t) =0, \lim_{t \to \infty} \tilde{Z}_{i}(t) =0 $, and then the consensus tracking error $ \lim_{t \to \infty} \tilde{x}_{i}(t) =0 $.

As we know, each follower has access to a weighted linear combination of relative states between itself and its neighbors. The network measurement for follower $ i $ is synthesized into a single signal as
\begin{equation}\label{network measurement}
\xi_{i}(t)= \sum_{j=1}^{N} a_{ij}(x_{i}(t)-x_{j}(t))+a_{i0}(x_{i}(t)-x_{0}(t)), i \in \mathbb{F}
\end{equation}
where $ a_{ij} $ is the $ (i,j) $-th entry of adjacency matrix $ \mathcal{A} $ of graph $ \mathcal{G} $. Especially, $ a_{i0}=1 $ means the follower $ i $ can get information from the leader and cannot otherwise. 
By using relative state information, denote a signal similar to \eqref{network measurement} as 
\begin{equation}\label{varrho}
\begin{aligned}
\varrho_{i}(t)=& a_{i0} [ \, v_{i}(t) 
- \int_{t-\tau}^{t} e^{A(t -s)}(Bu_{i}(s) + Ee^{S\tau}\hat{w}_{i}(s))ds \, ] +\sum_{j=1}^{N} a_{ij} \{ \, v_{i}(t)-v_{j}(t)  - \int_{t-\tau}^{t} e^{A(t -s)}  [B(u_{i}(s) -u_{j}(s))\\&+ Ee^{S\tau}(\hat{w}_{i}(s)-\hat{w}_{j}(s))]ds \,  \}-e^{A\tau}\xi_{i}(t).
\end{aligned}
\end{equation}
It is easy to calculate $ \varrho_{i}(t)= a_{i0}\tilde{v}_{i}(t) +\sum_{j=1}^{N} a_{ij}(\tilde{v}_{i}(t)-\tilde{v}_{j}(t))=\sum_{j=1}^{N} l_{ij}\tilde{v}_{j}(t) $.
\begin{remark}\label{remark_varrho}
	The signal $ \varrho_{i}(t) $, which will be used in the control protocol design, only needs 
	the relative state information $ \xi_{i}(t) $, the adaptive observer state $ v_{j}(t) $, the stored history of control input $ u_{j}(t-\tau) $ and disturbance observer state $ \hat{w}_{j}(t-\tau) $ of its neighbor $ j, j \in \mathbb{F} $ via the communication topology $ \mathcal{G} $.
	
\end{remark}
The fully distributed adaptive ESO is designed as

\begin{equation}\label{eq:protocol}
\dot{\bar{Z}}_{i}(t)=\underbrace{\begin{bmatrix}
	A+BK_{1} & 0 \\
	0 & S 
	\end{bmatrix}}_{\bar{A}_{1}}\bar{Z}_{i}(t)+\underbrace{\begin{bmatrix}
	K \\
	K^{'} 
	\end{bmatrix}}_{\bar{A}_{2}}(c_{i}(t)+\rho_{i}(t))\varrho_{i}(t)
\end{equation}
where $ K \in \mathbb{R}^{n \times n} $ and $ K^{'} \in \mathbb{R}^{s \times n} $ will be determined later. $ c_{i}(t) $ denotes the time-varying coupling weight associated with the $ i $-th follower and is used to make the whole controller fully distributed.
$ \rho_{i}(t) $ represents the smooth and nonnegative function. Both $ c_{i}(t) $ and $ \rho_{i}(t) $ are scalers and will be designed later.
From~\eqref{dot_Z_2} and \eqref{eq:protocol}, we have
\begin{equation*}
\begin{aligned}
\dot {\tilde{v}}_{i}(t) =& A\tilde{v}_{i}(t) + e^{A\tau}Ee^{S\tau}\tilde{w}_{i}(t-\tau) +K(c_{i}(t)+\rho_{i}(t))\varrho_{i}(t) ,\\
e^{S\tau}\dot{\tilde{w}}_{i}(t-\tau)=&Se^{S\tau}\tilde{w}_{i}(t-\tau)+e^{S\tau}K^{'}(c_{i}(t-\tau)+\rho_{i}(t-\tau))\varrho_{i}(t-\tau), \,i \in \mathbb{F}.
\end{aligned}
\end{equation*}
Denote $ e_{i}(t)=\begin{bmatrix}
\tilde{v}_{i}(t)\\ e^{S\tau}\tilde{w}_{i}(t-\tau)
\end{bmatrix}, \bar{K}=\begin{bmatrix}
K\\e^{S\tau}K^{'} 
\end{bmatrix} $. Note here that our objective is to prove $ \lim_{t \to \infty} \tilde{x}_{i}(t) =0 $, so it is equal to have $ c_{i}(t-\tau)=c_{i}(t), \rho_{i}(t-\tau)=\rho_{i}(t)$ and $ \varrho_{i}(t-\tau)=\varrho_{i}(t) $ when $ t \to \infty $, then
\begin{equation*}
\begin{aligned}
\dot e_{i}(t) = A_{T}e_{i}(t)+\bar{K}(c_{i}(t)+\rho_{i}(t))\sum_{j=1}^{N} l_{ij}Te_{i}(t)
\end{aligned}
\end{equation*}
where $ T=[I \, 0]\in  \mathbb{R}^{n \times (n+s)} $.
Similar to \eqref{network measurement} and \eqref{varrho},
denote a signal as
\begin{equation}\label{hat_e}
\hat{e}_{i}(t)=\sum_{j=1}^{N} l_{ij} e_{j}(t).
\end{equation}
The analysis of $ \hat{e}_{i}(t) $ is similar as Remark~\ref{remark_varrho}. Define $ \hat{e}(t)= [\hat{e}_{1}^{T}(t),\ldots,\hat{e}_{N}^{T}(t)]^{T} $, $ \hat{c}(t) = diag(c_{1}(t), \ldots, c_{N}(t)), \hat{\rho}(t) = diag(\rho_{1}(t),\ldots,\rho_{N}(t)) $ and $ e(t)= [e_{1}^{T}(t),\ldots,e_{N}^{T}(t)]^{T} $, then
\begin{equation}\label{nonautonomous_system_no}
\begin{aligned}
\dot{\hat{e}}(t)
=&(\mathcal{L}_{1}\otimes I_{n+s})\dot{e}(t)\\ 
=& [I_{N} \otimes A_{T} +\mathcal{L}_{1} (\hat{c}(t)+\hat{\rho}(t)) \otimes \bar{K}T]\hat{e}(t).\\
\end{aligned}
\end{equation}
The $ c_{i}(t) $ and $ \rho_{i}(t) $ are designed as follows
\begin{equation}\label{c_rho}
\begin{aligned}
\dot c_{i}(t) =& \hat{e}_{i}^{T}(t) \Gamma\hat{e}_{i}(t), \\ \rho_{i}(t)=&\hat{e}_{i}^{T}(t)P\hat{e}_{i}(t), \, i \in \mathbb{F}
\end{aligned}
\end{equation}
where $ c_{i}(0)\ge 0 $.
$ \Gamma \in  \mathbb{R}^{(n+s) \times (n+s)} $ and $ P \in  \mathbb{R}^{(n+s) \times (n+s)} $ are the feedback gain matrices to be determined in the following.

\begin{theorem}\label{theorem_disturbance}
	For the network-connected system with dynamics~\eqref{eq:dynamics} and \eqref{w_disturbance}, the fully distributed controller of~\eqref{input}, \eqref{eq:protocol} and \eqref{c_rho} solves the disturbance-rejecting consensus problem considering the control input time-delay under Assumptions \ref{assumption_tau}, \ref{assumptionABC} and \ref{assumptiondirected}
	if $ A+BK_{1} $ is Hurwitz, $ \Gamma = T^{T}T, \bar{K}=-P^{-1}T^{T} $ and $ P  > 0 $ is a solution to the LMI~\eqref{eq:ARE}.
	Moreover, the coupling weight $ c_{i}(t), i \in \mathbb{F} $ converge to some finite steady-state values.
\end{theorem}
\begin{proof}
	In the proof, we omit symbol $ (t) $ for convenience in writing if there is no special statements. 
	
	Let
	\begin{equation}\label{v1}
	V_{1}=\frac{1}{2}\sum_{i=1}^{N} g_{i}(2c_{i}+\rho_{i})\rho_{i}+\frac{1}{2}\sum_{i=1}^{N} g_{i}(c_{i}- \beta)^{2}
	\end{equation}
	where $ G = diag(g_{1}, \ldots, g_{N}) > 0 $ is a positive definite matrix such that $ G\mathcal{L}_{1}+\mathcal{L}_{1}^{T}G > 0 $. Since $ \mathcal{L}_{1} $ is a nonsingular $ M $-matrix, thus $ G $ exists based on Lemma \ref{lemma_M_matrix}. Particularly, $ g_{i}, i \in \textbf{I}[1,N] $ can be constructed as $ [g_{1}, \ldots, g_{N}]^{T}=(\mathcal{L}_{1}^{T})^{-1}[1,\ldots,1]^{T} $~\cite{wangchunyanCyber}.
	It is easy to get $ c_{i}(t) \ge 0, \forall t \ge 0 $ based on $\dot{c}_{i}(t) \ge 0, c_{i}(0)\ge 0 $ in~\eqref{c_rho}.
	$ \beta $ is a positive constant to be determined. Noting further that $ \rho_{i} \geq 0 $, so $ V_{1} $ is positive definite. Then
	\begin{equation}\label{eq:derivative_v1}
	\begin{aligned}
	\dot{V}_{1} 
	=& \displaystyle \sum_{i=1}^{N} [g_{i}(c_{i}+\rho_{i}) \dot{\rho}_{i} +g_{i}\rho_{i}\dot{c}_{i}+g_{i}(c_{i}-\beta)\dot{c}_{i}]\\
	=&\hat{e}^{T} [G(\hat{c}+\hat{\rho}) \otimes(PA_{T}+A_{T}^{T}P) +G(\hat{c}+\hat{\rho} - \beta I) \otimes \Gamma +(\hat{c}+\hat{\rho})(G\mathcal{L}_{1} + \mathcal{L}_{1}^{T}G) (\hat{c}+\hat{\rho}) \otimes P\bar{K}T]\hat{e}\\
	\le & \hat{e}^{T} [G(\hat{c}+\hat{\rho}) \otimes(PA_{T}+A_{T}^{T}P) + G(\hat{c}+\hat{\rho} - \beta I) \otimes T^{T}T -\lambda_{0}(\hat{c}+\hat{\rho})^{2} \otimes T^{T}T ] \hat{e}\\
	\end{aligned}
	\end{equation}
	where $ \lambda_{0}>0 $ is the smallest eigenvalue of $ G\mathcal{L}_{1} + \mathcal{L}_{1}^{T}G $. The inequality comes from $ G\mathcal{L}_{1} + \mathcal{L}_{1}^{T}G \ge \lambda_{0} I $ (Lemma \ref{lemma_M_matrix}). By using Lemma \ref{lemma_pq} we get
	
	\begin{equation}\label{inequality}
	\hat{e}^{T}[G(\hat{c}+\hat{\rho}) \otimes T^{T}T] \hat{e} \le \hat{e}^{T}[(\frac{\lambda_{0}}{2} (\hat{c}+\hat{\rho})^{2} + \frac{G^{2}}{2\lambda_{0}}) \otimes T^{T}T] \hat{e}.
	\end{equation}
	Substituting \eqref{inequality} into \eqref{eq:derivative_v1} yields
	\begin{equation}\label{v1_1}
	\begin{aligned}
	\dot{V}_{1} 
	\le & \hat{e}^{T} \{G(\hat{c}+\hat{\rho}) \otimes(PA_{T}+A_{T}^{T}P) -[\frac{\lambda_{0}}{2}(\hat{c}+\hat{\rho})^{2} -\frac{G^{2}}{2\lambda_{0}} +\beta G] \otimes T^{T}T \} \hat{e} \\
	\le & \hat{e}^{T} [G(\hat{c}+\hat{\rho}) \otimes(PA_{T}+A_{T}^{T}P-2T^{T}T) ] \hat{e}\\
	\le & 0.
	\end{aligned}
	\end{equation}
	Given
	the fact that $ a +b \geq 2 \sqrt{ab}, \forall a, b \in \mathbb{R}^{+} $, we have chosen $ \beta \geq \frac{5}{2\lambda_{0}}\max_{i \in \mathbb{F}} g_{i} $ to get the second inequality. The last
	inequality comes from LMI \eqref{eq:ARE}.
	
	So we can conclude that $ V_{1}(t) $ is bounded and so are $ \hat{e}_{i}$ and $ c_{i} $. It follows from \eqref{c_rho} and $ \Gamma = T^{T}T $ that $ \dot{c_{i}}(t) \ge 0 $, thus the coupling weights $  c_{i}(t), i \in \mathbb{F} $ increase monotonically and converge to some finite values finally, which verifies $ \lim_{t \to \infty} c_{i}(t-\tau)=\lim_{t \to \infty} c_{i}(t)  $. Note that $ \dot{V}_{1}(t) \equiv 0 $ is equivalent to $ \hat{e} =0 $. By LaSalle's Invariance principle \cite{krstic_nonlinear_1995}, it follows that $ \hat{e} $ asymptotically converges to zero, i.e., $ \lim_{t \to \infty}\hat{e} = 0 $. So from \eqref{c_rho}, $ \lim_{t \to \infty}\rho = 0 $  which verifies $ \lim_{t \to \infty}\rho(t-\tau) = \lim_{t \to \infty}\rho(t) $.
	
	Recalling that $ \hat{e} = (\mathcal{L}_{1} \otimes I_{n+s}) e $ in~\eqref{nonautonomous_system_no} and $ \mathcal{L}_{1} $ is nonsingular, we prove $ \lim_{t \to \infty} e=0 $. Considering $ e^{S\tau} $ is invertible, i.e., $ rank(e^{S\tau})=s $, we have $ \lim_{t \to \infty} \tilde{v}_{i}(t) =0, \lim_{t \to \infty} \tilde{w}_{i}(t)=0 $. Since $ \varrho=(\mathcal{L}_{1}\otimes I_{n})\tilde{v} $ from~\eqref{varrho}, it is easy to verify $ \lim_{t \to \infty}\varrho(t-\tau) = \lim_{t \to \infty}\varrho(t) $.
	
	Recall~\eqref{dot_Z_2} as
	\begin{equation}\label{Z_tilde}
	\dot{\tilde{Z}}_{i}=(A+BK_{1})\tilde{Z}_{i}+\tilde{K} e_{i} 
	\end{equation}
	where $ \tilde{K}=[BK_{1}, \, -e^{A\tau}E] $ and $ e_{i}=[\tilde{v}_{i}^{T}, \, (e^{S\tau}\tilde{w}_{i}(t-\tau))^{T}]^{T} $. 
	Since $ A+BK_{1} $ is Hurwitz and $ \lim_{t \to \infty} e_{i}=0 $, from \eqref{Z_tilde} we have $ \lim_{t \to \infty} \tilde{Z}_{i}=0, i \in \mathbb{F} $.
	
	Thanks to $ \lim_{t \to \infty} \tilde{Z}=0 $ and $ \lim_{t \to \infty} \tilde{v}_{i}=0 $, we have $\lim_{t \to \infty} v_{i}=0 $. As it is known that $ e^{A\tau} >0 $, 
	from~\eqref{most_important_equation}, we prove that the consensus tracking error $ \lim_{t \to \infty} \tilde{x}(t)=0 $, i.e., the proof is finished.
\end{proof}

\begin{remark}
	It is worth noting that for each follower $ i $, the variable $ \varrho_{i}(t) $ is very important for the fully distributed adaptive ESO design in~\eqref{eq:protocol}. The detailed explanation can be referred to Remark~\ref{remark_varrho}.
\end{remark}

\begin{remark}
	In contrast to the result \cite{wang_consensus_2017-1} where the consensus disturbance rejection problem of network-connected dynamic systems with input delay under undirected communication topology is solved,
	the distinctive feature of our whole control are twofolds: i) our controller is fully distributed; ii) the \ct{} is directed, which could save tremendous communication resources compared with the undirected topology.
\end{remark}

For the case there is no time-delay in the control input, we simply change $ \tilde{v}_{i}, e_{i} $ as $ \tilde{v}_{i}=v_{i}-\tilde{x}_{i}, e_{i}=[\tilde{v}_{i}^{T}, \tilde{w}_{i}^{T}]^{T} $, and modify the control input from~\eqref{input} to the following
\begin{equation}\label{input_nodelay}
u_{i}(t)=\begin{bmatrix}
K_{1} & -F 
\end{bmatrix}
\bar{Z}_{i}(t), i \in \mathbb{F}.
\end{equation} 
Then the consensus disturbance rejection problem under Assumptions \ref{assumptionABC} and \ref{assumptiondirected} is solved with the controller of~\eqref{input_nodelay}, \eqref{eq:protocol} and \eqref{c_rho}. Specifically, \eqref{nonautonomous_system_no} changes to
\begin{equation*}
\begin{aligned}
\dot{\hat{e}}(t)=& [I_{N} \otimes A_{T}^{'} +\mathcal{L}_{1} (\hat{c}(t)+\hat{\rho}(t)) \otimes \bar{K}^{'} T]\hat{e}(t)
\end{aligned}
\end{equation*}
where $ A_{T}^{'}=\begin{bmatrix}
A & E \\
0 & S 
\end{bmatrix}, \bar{K}^{'}=\begin{bmatrix}
K\\K^{'} 
\end{bmatrix} $ and $ T=\begin{bmatrix}
I & 0  
\end{bmatrix} $. From Lemma 1 of~\cite{ding2015consensus} it is known that 
$ (A_{T}^{'},T) $ is observable. The other parameters can be calculated similarly as the proof of Theorem~\ref{theorem_disturbance} and the detail is omitted here.

\subsection{Consensus tracking control with $ u_{0}(t) $}\label{subsec_consensus_with_leader_input}

In this subsection,the consensus tracking problem with leader's control input satisfying Assumption~\ref{leader input} is investigated. Correspondingly, \eqref{x_tilde} and \eqref{z_transformed_variable} change to
\begin{equation}\label{z_transformed_u0}
\begin{aligned}
\dot{\tilde{x}}_{i}(t)  =& A\tilde{x}_{i}(t)+B(u_{i}(t-\tau)-u_{0}(t-\tau))+E\bar{w}_{i}(t),\\
\tilde{Z}_{i}(t)=& e^{A\tau}\tilde{x}_{i}(t) + \int_{t-\tau}^{t} e^{A(t -s)}[B (u_{i}(s)-u_{0}(s)) + Ee^{S\tau}\hat{w}_{i}(s)]ds.
\end{aligned}
\end{equation}

Considering the leader's bounded input $ u_{0}(t) $, the following continuous nonlinear function $ z(\cdot) $ 
\begin{equation}\label{z_discontinuous}
z_{i}(x) =
\begin{cases}
\frac{x}{\|x\|}  & \text{if $ \|x\| > \sigma_{i}$,} \\
\frac{x}{\sigma_{i}} & \text{if $ \|x\| \le \sigma_{i}$}
\end{cases}
\end{equation}
is used to compensate the leader's input effect to the whole cooperative system. 
So based on~\eqref{input}, the modified control input is designed as
\begin{equation}\label{input_u0}
u_{i}(t)=(\bar{K}_{1}-\bar{F})\bar{Z}_{i}(t)  - \alpha z(\zeta_{i}(t)), i \in \mathbb{F}
\end{equation} 
such that $ \tilde{Z}_{i}(t) $ in~\eqref{z_transformed_u0} changes to
\begin{equation}\label{dot_Z_2_u0}
\begin{aligned}
\dot{\tilde{Z}}_{i}(t)=&(A+BK_{1})\tilde{Z}_{i}(t)+BK_{1}\tilde{v}_{i}(t)-e^{A\tau}Ee^{S\tau} \tilde{w}_{i}(t-\tau)
-B(\alpha z(\zeta_{i}(t))+u_{0}(t))
\end{aligned}
\end{equation}
where $ \alpha, \zeta_{i}(t) $ will be designed later. The ESO $ \bar{Z}_{i}(t)=[v_{i}(t)^{T}, \hat{w}_{i}(t)^{T}]^{T} $ is modified as

\begin{equation}\label{eq:protocol_u0}
\begin{aligned}
\dot{\bar{Z}}_{i}(t) =& \bar{A}_{1}\bar{Z}_{i}(t)+\bar{A}_{2}(c_{i}(t)+\rho_{i}(t))\varrho_{i}(t)-\bar{B} \alpha[z(\zeta_{i}(t))+z(\tilde{\zeta}_{i}(t))] ,\\
\dot c_{i}(t) =& \hat{e}_{i}^{T}(t) \Gamma\hat{e}_{i}(t) -\epsilon_{i}(c_{i}(t)-\beta_{1}), \\ \rho_{i}(t)=&\hat{e}_{i}^{T}(t)P\hat{e}_{i}(t), \, i \in \mathbb{F}
\end{aligned}
\end{equation}
where $ c_{i}(0)\ge \beta_{1} $ and $ \beta_{1}, \epsilon_{i} $ are positive constants. $ \tilde{\zeta}_{i}(t) $ will be designed later. Other variable formats are the same as in~\ref{subsec_consensus_without_leader_input} and \eqref{nonautonomous_system_no} changes to the following nonautonomous system $ \dot{\hat{e}}(t) = f( \hat{e}(t), t) $ as 
\begin{equation}\label{nonautonomous_system}
\begin{aligned}
\dot{\hat{e}}(t)= [I_{N} \otimes A_{T} +\mathcal{L}_{1} (\hat{c}(t)+\hat{\rho}(t)) \otimes \bar{K}T]\hat{e}(t)  - (\mathcal{L}_{1}\otimes \bar{B}) [\alpha z(\tilde{\zeta}(t)) -\textbf{1}\otimes u_{0}(t)].\\
\end{aligned}
\end{equation}
\begin{remark}
	From $ \tilde{Z}_{i}(t) $ in~\eqref{z_transformed_u0} and $ \varrho_{i}(t)= a_{i0}\tilde{v}_{i}(t) +\sum_{j=1}^{N} a_{ij}(\tilde{v}_{i}(t)-\tilde{v}_{j}(t)) $ with $ \tilde{v}_{i}(t)=v_{i}(t)-\tilde{Z}_{i}(t) $, we can see only a subset of followers need the historical information of leader's control input, i.e., $ u_{0}(t-\tau) $.
\end{remark}
\begin{theorem}\label{theorem_disturbance_u0}
	For the network-connected system with dynamics~\eqref{eq:dynamics} and \eqref{w_disturbance}, the fully distributed controller of~\eqref{input_u0} and \eqref{eq:protocol_u0} solves the consensus disturbance attenuation problem considering the input delay with the leader of bounded input under Assumptions \ref{assumption_tau}-\ref{assumptiondirected}
	if $ A+BK_{1} $ is Hurwitz, $ \Gamma = T^{T}T, \bar{K}=-P^{-1}T^{T} $, $ \alpha \ge \epsilon, \zeta _{i} (t) =B^{T}Q \tilde{Z}_{i}(t),  \tilde{\zeta}_{i}(t)=\bar{B}^{T}P\hat{e}_{i}(t) $
	and $ P \ge 0, Q \ge 0 $ are solutions to the following LMIs
	\begin{align}
	PA_{T}+A_{T}^{T}P+\mu P-2T^{T}T  <0, \label{lmi}\\
	Q(A+BK_{1}) +(A+BK_{1})^{T}Q<0 \label{lmi_Q}
	\end{align}
	where $ \mu >1 $. The consensus tracking error $ \tilde{x}_{i}(t) $ converges exponentially to the residual set
	\begin{equation}\label{tilde_x_intervel}
	\Pi= \left\lbrace    \tilde{x}_{i}(t): \| \tilde{x}_{i}(t) \| \le \|\tilde{Z}_{i}(t-\tau)\| + \chi \|E\| \, \|e^{S\tau}\tilde{w}_{i}(t-\tau) \|  \right\rbrace 
	\end{equation}
	where $ \chi = \| \int_{-\tau}^{0} e^{As}ds \| $. 
	$ \tilde{Z}_{i}(t), \tilde{w}_{i}(t) $ satisfy~\eqref{Z_mode} and \eqref{e_interval} in the proof, respectively.
	Besides, $ c_{i}(t), i \in \mathbb{F} $ are uniformly ultimately bounded.
\end{theorem}
\begin{proof}
	In the proof, we omit symbol $ (t) $ for convenience in writing if there is no special statements.
	The Lyapunov function candidate is the same as~\eqref{v1}, and after the same calculation as in the proof of Theorem~\ref{theorem_disturbance}, \eqref{v1_1} changes to
	\begin{equation}\label{v1_1_u0}
	\dot{V}_{1} 
	\le  -\hat{e}^{T} [G(\hat{c}+\hat{\rho}) \otimes \mathcal{H}] \hat{e} -  \sum_{i=1}^{N} g_{i}(c_{i}-\beta)\epsilon_{i}(c_{i}-\beta_{1}) +\Omega
	\end{equation}
	where $ \mathcal{H}=-(PA_{T}+A_{T}^{T}P-2T^{T}T) >0 $ and $ \Omega=-2  \hat{e}^{T}[G(\hat{c}+\hat{\rho})\mathcal{L}_{1} \otimes P\bar{B}][\alpha z(\tilde{\zeta}(t)) -\textbf{1}\otimes u_{0}(t)] $.
	
	Firstly, we come to deal with the leader's bounded input $ u_{0}(t) $ and the nonlinear function $ z(\cdot) $ in $ \Omega $. Using the Laplacian matrix property $ \mathcal{L}_{1}\textbf{1}=-\mathcal{L}_{2} $ and Assumption~\ref{leader input}, we get
	\begin{equation}\label{z2}
	\begin{aligned}
	\hat{e}^{T}[G(\hat{c}+\hat{\rho})\mathcal{L}_{1}\otimes P\bar{B}] (\textbf{1} \otimes u_{0}(t)) &= \sum_{i=1}^{N} [ g_{i}(c_{i}+\rho_{i}) \hat{e}_{i}^{T} P\bar{B}a_{i0} u_{0}(t) ]
	\\ & \le  \sum_{i=1}^{N}g_{i}(c_{i}+\rho_{i}) \|\bar{B}^{T}P\hat{e}_{i}\| a_{i0} \epsilon.
	\end{aligned}
	\end{equation}
	On the other hand, considering the following three cases.
	
	i) $ \|\bar{B}^{T}P\hat{e}_{i}\| > \sigma_{i}, i\in \mathbb{F} $, then
	\begin{equation*}
	\begin{aligned}
	\hat{e}_{i}^{T}P\bar{B}z(\bar{B}^{T}P\hat{e}_{i})=&\hat{e}_{i}^{T}P\bar{B} \frac{\bar{B}^{T}P\hat{e}_{i}}{\|\bar{B}^{T}P\hat{e}_{i}\|}=\|\bar{B}^{T}P\hat{e}_{i}\|, \\
	\hat{e}_{i}^{T}P\bar{B}z(\bar{B}^{T}P\hat{e}_{j})\le&\|\hat{e}_{i}^{T}P\bar{B}\| \left\| \frac{\bar{B}^{T}P\hat{e}_{j}}{\|\bar{B}^{T}P\hat{e}_{j}\|} \right\| =\|\bar{B}^{T}P\hat{e}_{i}\|.
	\end{aligned}
	\end{equation*}
	Here is the reason we choose $ \tilde{\zeta}_{i}(t)=\bar{B}^{T}P\hat{e}_{i} $, then
	\begin{equation}\label{z1}
	\begin{aligned}
	-\hat{e}^{T}[G(\hat{c}+\hat{\rho})\mathcal{L}_{1}\otimes P\bar{B} ]\alpha z(\tilde{\zeta}(t))
	=& -\sum_{i=1}^{N} \{  g_{i}(c_{i}+\rho_{i}) \alpha \hat{e}_{i}^{T} P\bar{B} [a_{i0}z(\bar{B}^{T}P\hat{e}_{i})+ \sum_{j=1}^{N}a_{ij}(z(\bar{B}^{T}P\hat{e}_{i})-z(\bar{B}^{T}P\hat{e}_{j}) )] \} \\
	\le & -\sum_{i=1}^{N}g_{i}(c_{i}+\rho_{i})  \|\bar{B}^{T}P\hat{e}_{i}\| a_{i0} \alpha.
	\end{aligned}
	\end{equation}
	Combining \eqref{z1} and \eqref{z2} with $ \alpha \ge \epsilon $, we have
	\begin{equation}\label{inequality_uo_z}
	\Omega \le 0.
	\end{equation}
	
	ii) $ \|\bar{B}^{T}P\hat{e}_{i}\| \le \sigma_{i}, i\in \mathbb{F} $, then
	\begin{equation*}
	\|z(\bar{B}^{T}P\hat{e}_{j})\|= \|\frac{\bar{B}^{T}P\hat{e}_{j}}{\sigma_{j}} \|\le 1.
	\end{equation*}
	Due to $ a_{ij} = 0 / 1 $ in $ \mathcal{A} $ of the graph $ \mathcal{G} $, we get
	\begin{equation}\label{inequality_uo_z1}
	\begin{aligned}
	\Omega & \le \sum_{i=1}^{N}g_{i}[2(c_{i}-\beta_{1})+2\rho_{i}+2\beta_{1}]   [a_{i0} \epsilon+(2N-1)\alpha]\sigma_{i}\\& \le \sum_{i=1}^{N} \frac{g_{i}\epsilon_{i}}{4}(c_{i}-\beta_{1})^{2} + \sum_{i=1}^{N}\frac{\lambda_{min}(\mathcal{H})}{2\lambda_{max}(P)}g_{i}\rho_{i}^{2}+ \Xi_{1}
	\end{aligned}
	\end{equation}
	where 
	\begin{equation}\label{xi_1}
	\begin{aligned}
	\Xi_{1}=& \frac{(\beta-\beta_{1})^{2}}{2} \sum_{i=1}^{N} g_{i}\epsilon_{i}       + \sum_{i=1}^{N} g_{i}\sigma_{i}[a_{i0} \epsilon+(2N-1)\alpha] \{ 2\beta_{1} +(\frac{4}{\epsilon_{i}} + \frac{2\lambda_{max}(P)}{\lambda_{min}(\mathcal{H})})\sigma_{i}[a_{i0} \epsilon+(2N-1)\alpha] \}.
	\end{aligned}
	\end{equation}

	iii) $ \hat{e}_{i}, i\in \mathbb{F} $ satisfy neither case i) nor case ii). Generally, assume $ \|\bar{B}^{T}P\hat{e}_{i}\| > \sigma_{i}, i=1,\ldots,k $, and $ \|\bar{B}^{T}P\hat{e}_{i}\| \le \sigma_{i}, i=k+1,\ldots,N $, then 
	\begin{equation}\label{inequality_uo_z2}
	\Omega\le 2\sum_{i=k+1}^{N}g_{i}(c_{i}+\rho_{i})  [a_{i0} \epsilon+(2N-1)\alpha]\sigma_{i}.
	\end{equation}
	Comparing \eqref{inequality_uo_z}, \eqref{inequality_uo_z1} and \eqref{inequality_uo_z2}, we find out that $ \Omega $ satisfies~\eqref{inequality_uo_z1}.
	Note that
	\begin{equation*}
	\begin{aligned}
	-(c_{i}-\beta)(c_{i}-\beta_{1}) =& -(c_{i}-\beta)^{2}-(c_{i}-\beta)(\beta -\beta_{1})\\
	\le&-\frac{1}{2}(c_{i}-\beta)^{2}+\frac{1}{2}(\beta-\beta_{1})^{2}
	\end{aligned}
	\end{equation*}
	and
	\begin{equation*}
	\begin{aligned}
	-(c_{i}-\beta)(c_{i}-\beta_{1}) =& -(c_{i}-\beta_{1})^{2}-(\beta_{1}-\beta)(c_{i} -\beta_{1})\\
	\le&-\frac{1}{2}(c_{i}-\beta_{1})^{2}+\frac{1}{2}(\beta-\beta_{1})^{2}.
	\end{aligned}
	\end{equation*}
	Then substituting above two inequalities and \eqref{inequality_uo_z1} into~\eqref{v1_1_u0}, we obtain
	\begin{equation}\label{V_1}
	\dot{V}_{1} 
	\le  -\frac{1}{2}\hat{e}^{T} [G(\hat{c}+\hat{\rho}) \otimes \mathcal{H} ] \hat{e}-\sum_{i=1}^{N} \frac{g_{i}\epsilon_{i}}{4} (c_{i}-\beta)^{2}+\Xi_{1}.
	\end{equation}
	Thanks to $ \mu >1, P\ge 0 $ and the LMI~\eqref{lmi}, we have $ \mathcal{H}>(\mu -1)P\ge 0 $. What is more, $ \sum_{i=1}^{N} \frac{g_{i}\epsilon_{i}}{4}(c_{i}-\beta)^{2}  \ge 0 $, then	
	\begin{align}
	\dot{V}_{1} 
	\le&  -\frac{1}{2}\hat{e}^{T} [G(\hat{c}+\hat{\rho}) \otimes \mathcal{H} ] \hat{e}+\Xi_{1}\label{deri_v1_1}.
	\end{align}

	Define the continuous function $ \mathcal{K}_{3} (\| \hat{e} \|)=\hat{e}^{T} [G(\hat{c}+\hat{\rho}) \otimes \mathcal{H} ] \hat{e} $. Because of $ G(\hat{c}+\hat{\rho}) >0 $ and $ \mathcal{H}>0 $, it is easy to verify $ \mathcal{K}_{3} $ belongs to class $ \mathcal{K} $ function. Considering  $ \Xi_{1}>0 $, from~\eqref{deri_v1_1} and Lemma~\ref{lemma_kappa_function}, it is easy to conclude that $ \hat{e} (t) $, which is the solution of the nonautonomous system $ \dot {\hat{e}} = f( \hat{e}, t) $ in~\eqref{nonautonomous_system}, is uniformly ultimately bounded.

	Secondly, considering $ \rho_{i}\ge 0 $, from~\eqref{v1} we get
	\begin{equation}\label{kappaV1}
	\kappa_{1} V_{1}
	\le  \kappa_{1} \sum_{i=1}^{N} g_{i}(c_{i}+\rho_{i}) \hat{e}_{i}^{T}P\hat{e}_{i} + \sum_{i=1}^{N} \frac{\kappa_{1} g_{i}}{2}(c_{i}-\beta)^{2}
	\end{equation}
	where $ \kappa_{1} >0 $ is a small positive constant to be designed later. Combine~\eqref{V_1} and \eqref{kappaV1}, then
	\begin{equation}\label{dot_v1}
	\begin{aligned}
	\dot{V}_{1} 
	\le&   -\frac{1}{2}\hat{e}^{T} [G(\hat{c}+\hat{\rho}) \otimes \mathcal{H} ] \hat{e}-\sum_{i=1}^{N} \frac{g_{i}\epsilon_{i}}{4} (c_{i}-\beta)^{2}+\Xi_{1}  -\kappa_{1} V_{1}+\kappa_{1} \sum_{i=1}^{N} g_{i}(c_{i}+\rho_{i})\hat{e}_{i}^{T}P\hat{e}_{i} +\sum_{i=1}^{N}\frac{g_{i}\kappa_{1}}{2}(c_{i}-\beta)^{2}\\
	=&-\kappa_{1} V_{1}-\frac{1}{2}\hat{e}^{T} [G(\hat{c}+\hat{\rho}) \otimes (\mathcal{H}-2\kappa_{1} P )] \hat{e} -\sum_{i=1}^{N} \frac{g_{i}(\epsilon_{i}-2\kappa_{1})}{4}(c_{i}-\beta)^{2}+\Xi_{1}.
	\end{aligned}
	\end{equation}
	Define $ \mu = 1 + 2\kappa_{1} $, then $ \mathcal{H}-2\kappa_{1} P >0 $ based on the LMI~\eqref{lmi}. Choose $ 0<\kappa_{1} \le \min_{i \in \mathbb{F}} \frac{\epsilon_{i}}{2} $, then we obtain
	\begin{equation}
	\dot{V}_{1} \le -\kappa_{1} V_{1}+\Xi_{1}.
	\end{equation}
	In light of Lemma~\ref{lemma_dotV_V}, we could deduce that $ V_{1} $ exponentially converges to the residual set $ \Pi_{1}= \{V_{1}: V_{1}< \frac{\Xi_{1}}{\kappa_{1}}\}$ with a convergence rate faster than $ e^{-\kappa_{1} t} $. From \eqref{v1} we have $ V_{1}\ge \min_{i \in \mathbb{F}} g_{i} [ \lambda_{min}(P) \|\hat{e}\|^{2} + \frac{1}{2}\sum_{i=1}^{N} (c_{i}-\beta)^{2}] $. Since $ \hat{e} (t) $ is uniformly ultimately bounded and $ \beta \geq \frac{5}{2\lambda_{0}}\max_{i \in \mathbb{F}} g_{i} $ is a constant, we can conclude that $ c_{i}, i\in \mathbb{F} $ are uniformly ultimately bounded.
	
	Furthermore, note from~\eqref{dot_v1} that if $  \|\hat{e}\|^{2} > \frac{2\Xi_{1}}{\lambda_{min}(G)\lambda_{min}(\mathcal{H}-2\kappa_{1} P)} $, then $ \dot{V}_{1} \le -\kappa_{1} V_{1} $. Therefore, $ \hat{e} $ is uniformly ultimately bounded satisfying
	\begin{equation}\label{hat_e_interval}
	\|\hat{e}\|^{2} \le \frac{2\Xi_{1}}{\lambda_{min}(G)\lambda_{min}(\mathcal{H}-2\kappa_{1} P)}.
	\end{equation}
	
	Thirdly, recall~\eqref{dot_Z_2_u0} as
	\begin{equation}
	\dot{\tilde{Z}}_{i}=(A+BK_{1})\tilde{Z}_{i}+\tilde{K} e_{i} -B(\alpha z(\zeta_{i}(t))+u_{0}(t))
	\end{equation}
	where $ \tilde{K}=[BK_{1}, \, -e^{A\tau}E] $ and $ e_{i}=[\tilde{v}_{i}^{T}, \, (e^{S\tau}\tilde{w}_{i}(t-\tau))^{T}]^{T} $.
	
	Let
	\begin{equation}\label{V2}
	V_{2}=\sum_{i=1}^{N} \tilde{Z}_{i}^{T} Q\tilde{Z}_{i} +\gamma_{1} V_{1}
	\end{equation}
	where $ \gamma_{1} $ is a positive constant to be designed later. Then
	\begin{equation}\label{deriva_V2}
	\dot{V}_{2} = -\tilde{Z}^{T}(I_{N}\otimes \mathcal{X})\tilde{Z} +2\tilde{Z}^{T}(\mathcal{L}^{-1}_{1} \otimes Q\tilde{K})\hat{e}+\Omega_{1} +\gamma_{1} \dot{V}_{1}
	\end{equation}
	where $ \Omega_{1} = - 2\tilde{Z}^{T}(I_{N}\otimes QB) [\alpha z(\zeta (t)) +\textbf{1}\otimes u_{0}(t)]$ and $ \mathcal{X}= -[Q(A+BK_{1}) +(A+BK_{1})^{T}Q]>0 $ based on the LMI~\eqref{lmi_Q}.  Here is the reason that we design $ \zeta _{i} (t) =B^{T}Q\tilde{Z}_{i}(t) $. We omit the detail which is similar as  \eqref{inequality_uo_z}, \eqref{inequality_uo_z1} and \eqref{inequality_uo_z2}. It is worth noting that when $ \|B^{T}Q\tilde{Z}_{i}(t)\| \le \sigma_{i}, i\in \mathbb{F} $, then
	\begin{equation*}
	- 2\tilde{Z}^{T}(I_{N}\otimes QB) \alpha z(\zeta (t)) =-2\alpha \sum_{i=1}^{N}\frac{\tilde{Z}_{i}^{T}QBB^{T}Q\tilde{Z}_{i}}{\sigma_{i}} \le 0.
	\end{equation*}
	Then it is easy to get
	\begin{equation}
	\Omega_{1} \le  2\epsilon \sum_{i=1}^{N}\sigma_{i}.
	\end{equation}
	By using Lemma~\ref{lemma_pq} we have 
	\begin{equation*}
	\begin{aligned}
	2\tilde{Z}^{T}(\mathcal{L}^{-1}_{1}\otimes Q\tilde{K})\hat{e}
	\le& \frac{1}{2}\tilde{Z}^{T}(I_{N}\otimes \mathcal{X})\tilde{Z} +\frac{2\lambda_{max}(\tilde{K}^{T}QQ\tilde{K})}{\lambda^{2}_{min}(\mathcal{L}_{1})\lambda_{min}(\mathcal{X})}\hat{e}^{T}\hat{e},
	\end{aligned}
	\end{equation*} 
	Denote 
	\begin{equation}\label{xi_2}
	\Xi_{2}=\gamma_{1}\Xi_{1} +2\epsilon \sum_{i=1}^{N}\sigma_{i}.
	\end{equation}
	Then substituting the above inequality and \eqref{V_1} into~\eqref{deriva_V2} gives
	\begin{equation}\label{deriv_V2_1}
	\begin{aligned}
	\dot{V}_{2} 
	\le & 	-\frac{1}{2}\tilde{Z}^{T}(I_{N}\otimes \mathcal{X})Z + \frac{2\lambda_{max}(\tilde{K}^{T}QQ\tilde{K})}{\lambda^{2}_{min}(\mathcal{L}_{1})\lambda_{min}(\mathcal{X})}\hat{e}^{T}\hat{e} -\frac{1}{2}\hat{e}^{T} [G(\hat{c}+\hat{\rho}) \otimes (\gamma_{1}-2) \mathcal{H} ] \hat{e} -\hat{e}^{T} [G(\hat{c}+\hat{\rho}) \otimes \mathcal{H} ] \hat{e}\\
	&-\gamma_{1} \sum_{i=1}^{N} \frac{g_{i}\epsilon_{i}}{4}(c_{i}-\beta)^{2}+\Xi_{2} - \kappa_{2}V_{2}+ \kappa_{2}V_{2}.
	\end{aligned}
	\end{equation}
	Since $ c_{i}(t) \ge \beta_{1}, \forall t\ge0 $, here we design $ \beta_{1}\ge 1 $ such that $ (\hat{c}+\hat{\rho}) \ge I $.
	Let $ \gamma_{1} \ge 2 $ temporarily, then $$ \hat{e}^{T} [G(\hat{c}+\hat{\rho}) \otimes (\gamma_{1}-2) \mathcal{H} ] \hat{e} \ge (\gamma_{1}-2)\lambda_{min}(G)\lambda_{min}(\mathcal{H})\hat{e}^{T}\hat{e}. $$
	Now choose $ \gamma_{1} \ge \frac{4\lambda_{max}(\tilde{K}^{T}QQ\tilde{K})}{\lambda^{2}_{min}(\mathcal{L}_{1})\lambda_{min}(\mathcal{X})\lambda_{min}(G)\lambda_{min}(\mathcal{H})} +2$ with calculating $ \kappa_{2}V_{2} $ similarly as~\eqref{kappaV1} such that \eqref{deriv_V2_1} turns to
	\begin{equation}\label{V2_2}
	\begin{aligned}
	\dot{V}_{2} \le & - \kappa_{2}V_{2}+\Xi_{2}-\frac{1}{2}\tilde{Z}^{T}[I_{N}\otimes (\mathcal{X}-2\kappa_{2} Q)]\tilde{Z} -\hat{e}^{T} [G(\hat{c}+\hat{\rho}) \otimes (\mathcal{H}-\gamma_{1} \kappa_{2}P) ] \hat{e} -\frac{\gamma_{1}g_{i}}{4} \sum_{i=1}^{N} (\epsilon_{i}-2\kappa_{2})(c_{i}-\beta)^{2}.
	\end{aligned}
	\end{equation}
	By choosing $ \kappa_{2}=\min \{\frac{\lambda_{min}(\mathcal{X})}{2\lambda_{max}(Q)}, \frac{\lambda_{min}(\mathcal{H})}{\gamma_{1}\lambda_{max}(P)}, \min_{i \in \mathbb{F}}\frac{\epsilon_{i}}{2}  \} $,
	we have 
	\begin{equation}
	\dot{V}_{2} \le -\kappa_{2} V_{2}+\Xi_{2}.
	\end{equation}
	Similar as the proof of boundedness of $ V_{1} $, based on the Lemma~\ref{lemma_dotV_V}, it can be deduced that $ V_{2} $ exponentially converges to the residual set $ \Pi_{2}= \{V_{2}: V_{2}< \frac{\Xi_{2}}{\kappa_{2}}\}$ with a convergence rate faster than $ e^{-\kappa_{2} t} $. From \eqref{V2} we have $ V_{2}\ge \lambda_{min}(Q)\|\tilde{Z}\|^{2}+ \gamma_{1} V_{1} $. Since $ V_{1} $ is uniformly ultimately bounded and $ \gamma_{1} $ is a chosen positive constant, we can conclude that $ \tilde{Z} $ is uniformly ultimately bounded 
	with $  \|\tilde{Z}\|^{2} \le \frac{\Xi_{2}}{\kappa_{2}\lambda_{min}(Q)} $.
	
	Furthermore, note from~\eqref{V2_2} that if $  \|\tilde{Z}\|^{2} > \frac{2\Xi_{2}}{\lambda_{min}(\mathcal{X}-2\kappa_{2} Q)} $, then $ \dot{V}_{2} \le -\kappa_{2} V_{2} $. Therefore, $ \tilde{Z} $ is uniformly ultimately bounded satisfying
	\begin{equation}\label{Z_mode}
	\|\tilde{Z}\| \le \sqrt{\min\{ \frac{\Xi_{2}}{\kappa_{2}\lambda_{min}(Q)}, \frac{2\Xi_{2}}{\lambda_{min}(\mathcal{X}-2\kappa_{2} Q)}\}}.
	\end{equation}
	If $  \|\hat{e}\|^{2} > \frac{\Xi_{2}}{\lambda_{min}(G)\lambda_{min}(\mathcal{H}-\gamma_{1} \kappa_{2}  P)} $, then $ \dot{V}_{2} \le -\kappa_{2} V_{2} $. Combined with~\eqref{hat_e_interval}, we have
	\begin{equation}\label{hat_e_interval_2}
	\|\hat{e}\|^{2} \le \min\{ \frac{2\Xi_{1}}{\lambda_{min}(G)\lambda_{min}(\mathcal{H}-2\kappa_{1} P)}, \frac{\Xi_{2}}{\lambda_{min}(G)\lambda_{min}(\mathcal{H}-\gamma_{1} \kappa_{2}  P)}\}.
	\end{equation}
	From~\eqref{hat_e} where $ e=(\mathcal{L}_{1}^{-1}\otimes I) \hat{e} $, we have $ \|e\| \le \frac{\| \hat{e}\|}{\lambda_{min}(\mathcal{L}_{1})} $. In addition to $ e_{i}=[\tilde{v}_{i}^{T}, \, (e^{S\tau}\tilde{w}_{i}(t-\tau))^{T}]^{T} $, we come to conclusion that $ e^{S\tau}\tilde{w}(t-\tau) $ is uniformly ultimately bounded satisfying
	\begin{equation}\label{e_interval}
	\|e^{S\tau}\tilde{w}(t-\tau) \| \le \frac{\| \hat{e}\|}{\lambda_{min}(\mathcal{L}_{1})}
	\end{equation}
	where $ \hat{e} $ satisfies~\eqref{hat_e_interval_2}.
	
	Fourthly,
	the exact prediction at time $ t $ of the consensus tracking error $ \tilde{x}_{i}(t) $ of the system~\eqref{z_transformed_u0} at time $ t+\tau $ is 
	\begin{equation*}
	x_{pi}(t)=e^{A\tau}\tilde{x}_{i}(t) + \int_{t-\tau}^{t} e^{A(t -s)}[B(u_{i}(s)-u_{0}(s))+E\bar{w}_{i}(s+\tau)]ds
	\end{equation*}
	for all $ t\ge 0 $, which, in other words, $ x_{pi}(t)=\tilde{x}_{i}(t+\tau) $. Similarly, $ \tilde{Z}_{i}(t) $ in~\eqref{z_transformed_u0} estimate $ \tilde{x}_{i}(t+\tau) $, and the estimating error is
	\begin{equation}\label{z_transformed_variable_error}
	\tilde{x}_{i}(t)-\tilde{Z}_{i}(t-\tau)=x_{pi}(t-\tau)-\tilde{Z}_{i}(t-\tau)=-\int_{t-\tau}^{t} e^{A(t -s)}Ee^{S\tau}\tilde{w}_{i}(s)ds.
	\end{equation}
	Then we conclude that the consensus tracking error $ \tilde{x}_{i}(t) $ converges exponentially to the residual set $ \Pi $ in Theorem~\ref{theorem_disturbance_u0}.
\end{proof}

\begin{remark}
	For $ u_{i}(t-\tau) $ of follower $ i $ in~\eqref{input_u0}, $ \tilde{Z}_{i}(t-\tau), j \in \mathbb{F} $ in~\eqref{z_transformed_u0} is needed in calculation of $ z(\zeta_{i}(t-\tau)) $. Since $ \tilde{Z}_{i}(t-\tau) $ is not defined for $ t \in [0,\, \tau] $, set $ \tilde{Z}_{i}(t-\tau)=\tilde{Z}_{i}(0) $ for all $ t \in [0,\, \tau] $.
\end{remark}

\begin{remark}\label{remark_tuning_parameters}
	From~\eqref{xi_1}, the value of $ \Xi_{1} $ is proportional to the upper bound $ \epsilon $ of leader's input, $ \alpha $ satisfying $ \alpha \ge \epsilon $, $ \sigma_{i} $ in~\eqref{z_discontinuous}, the followers' number $ N $, and $ a_{i0} $ which means how many followers can receive the leader's information. Then from~\eqref{tilde_x_intervel}, \eqref{Z_mode}-\eqref{e_interval} and \eqref{xi_2}, the upper bound of consensus tracking error $ \tilde{x}_{i}(t) $ can be controlled to be small by tuning the above parameters.
\end{remark}

\section{Simulation}\label{simulaiton}

\textbf{Example 1.}
This example verifies Theorem~\ref{theorem_disturbance}. Consider system~\eqref{eq:dynamics} and~\eqref{w_disturbance} with
$$
A =  
\begin{bmatrix}
-4 & 1\\
1 & 0
\end{bmatrix},
B =\begin{bmatrix}
1 & 2\\
2 & 1
\end{bmatrix},
S =\begin{bmatrix}
0 & 1\\
-1 & 0
\end{bmatrix}
$$
and $ F=I_{2}, E=BF $. Then $ (A,B) $ is controllable and $ (S,E) $ is observable. $ \lambda_{1}(A)=-4.2361 $ and $\lambda_{2}(A)= 0.2361 $ means that our fully distributed controller can be applied to any open loop linear MASs with constant input delay and disturbances.
The communication topology $ \mathcal{G} $ is shown in Fig. \ref{communication topology} satisfying Assumption~\ref{assumptiondirected}. 	
Solving LMI~\eqref{eq:ARE} gets\\
\begin{equation*}
P=\begin{bmatrix}
0.3554 &   0.0230  & -0.1985 &  -0.0195\\
0.0230  &  0.5864 &  -0.7986  &  0.0854\\
-0.1985 &  -0.7986 &   3.5022 &  -0.7468\\
-0.0195 &   0.0854  & -0.7468 &   2.4724
\end{bmatrix},
\end{equation*} 
and then other parameters can be calculated accordingly.
Using pole placement method, assign eigenvalues of $ A+BK_{1} $ as -5,-10 and get $ K_{1}=[   -0.3333 ,  -6.3333;
-0.3333 ,   2.6667
] $. Similarly, when there is no input delay, the solution to LMI $ P^{'}A_{T}^{'}+A_{T}^{'T}P^{'}-2T^{T}T < 0 $ is
\begin{equation*}
P^{'}=\begin{bmatrix}
0.3337  &  0.0200 &  -0.2022  & -0.0351\\
0.0200  &  0.6059 &  -0.7971  &  0.1013\\
-0.2022 &  -0.7971 &   3.4524 &  -0.7308\\
-0.0351 &   0.1013 &  -0.7308 &   2.4451
\end{bmatrix}.
\end{equation*} 

Set the initial states as $ x_{ij}(0)= 4\delta +1, w_{ij}(0)= 10\delta -5$, $ c_{i}(0)=4\delta +1, i\in \mathbb{F} $ and $ x_{0j}(0)= 3\delta +5,w_{0j}(0)= 3\delta +1, j\in \textbf{I}[1,2]$, where $ \delta $ is a pseudorandom value with a uniform distribution on the interval $ (0,1) $. The input delay is taken as $ \tau=0.09s $ and $ u(t)=0, \forall t \in [-\tau, 0] $. 
\begin{remark}
	Compared with the values of initial states, the values of disturbances are quite large.
\end{remark}

Fig.~\ref{fig:comparison} shows the comparison result under the same initial conditions without input delay and with input delay, respectively. The consensus tracking errors are illustrated in Fig.~\ref{fig:tiled_x_without_delay} and~\ref{fig:tiled_x_with_delay} where the delay effect is well compensated. It can be seen from Fig.~\ref{fig:input_without_delay} and~\ref{fig:input_with_delay} that at the beginning the delayed system needs larger control input. Fig.~\ref{fig:observer_error} and~\ref{fig:disturbance_observer_error} present the ESO $ \bar{Z}_{i}(t)=[v_{i}(t)^{T}, \hat{w}_{i}(t)^{T}]^{T} $ tracking errors which state clearly the effectiveness of fully distributed adaptive ESO. Particularly, Fig.~\ref{fig:ci_varrho_rho} verifies the assumption that $ \lim_{t \to \infty} c_{i}(t-\tau)=c_{i}(t), \lim_{t \to \infty} \rho_{i}(t-\tau)=\rho_{i}(t)$ and $ \lim_{t \to \infty} \varrho_{i}(t-\tau)=\varrho_{i}(t), i\in \mathbb{F} $.
\begin{figure}[!tb]
	\centering
	\includegraphics[width=0.5\hsize]{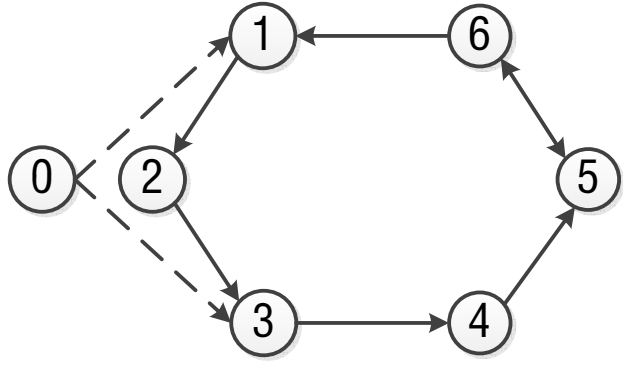}
	\caption{The directed communication topology $ \mathcal{G} $.}
	\label{communication topology}
\end{figure}

\begin{figure*}
	\centering
	\begin{subfigure}[h]{0.47\textwidth}
		\includegraphics[width=\textwidth]{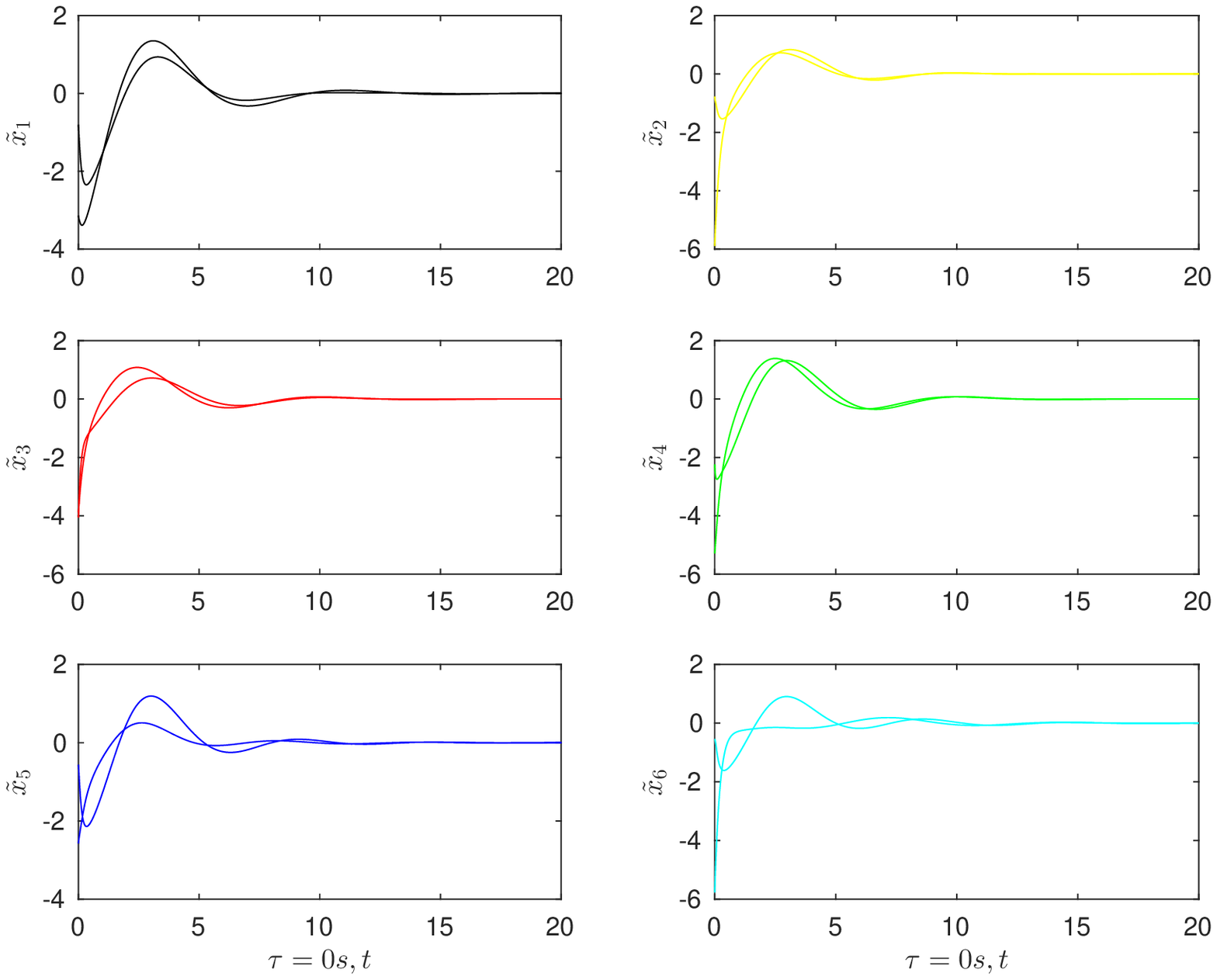}
		\caption{Consensus tracking error without delay.}
		\label{fig:tiled_x_without_delay}
	\end{subfigure}
	~
	\begin{subfigure}[h]{0.47\textwidth}
		\includegraphics[width=\textwidth]{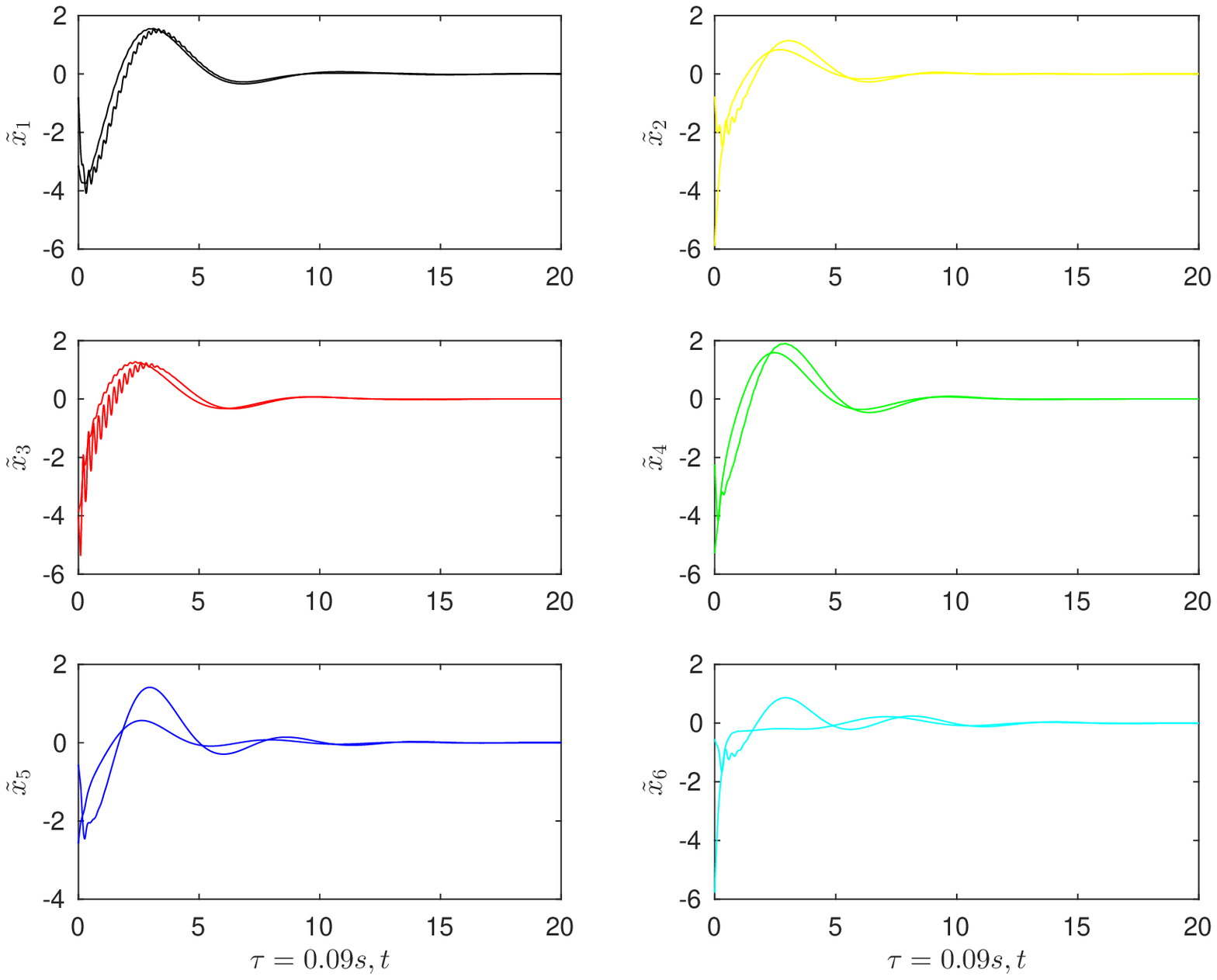}
		\caption{Consensus tracking error with delay.}
		\label{fig:tiled_x_with_delay}
	\end{subfigure}
	~
	\begin{subfigure}[h]{0.47\textwidth}
		\includegraphics[width=\textwidth]{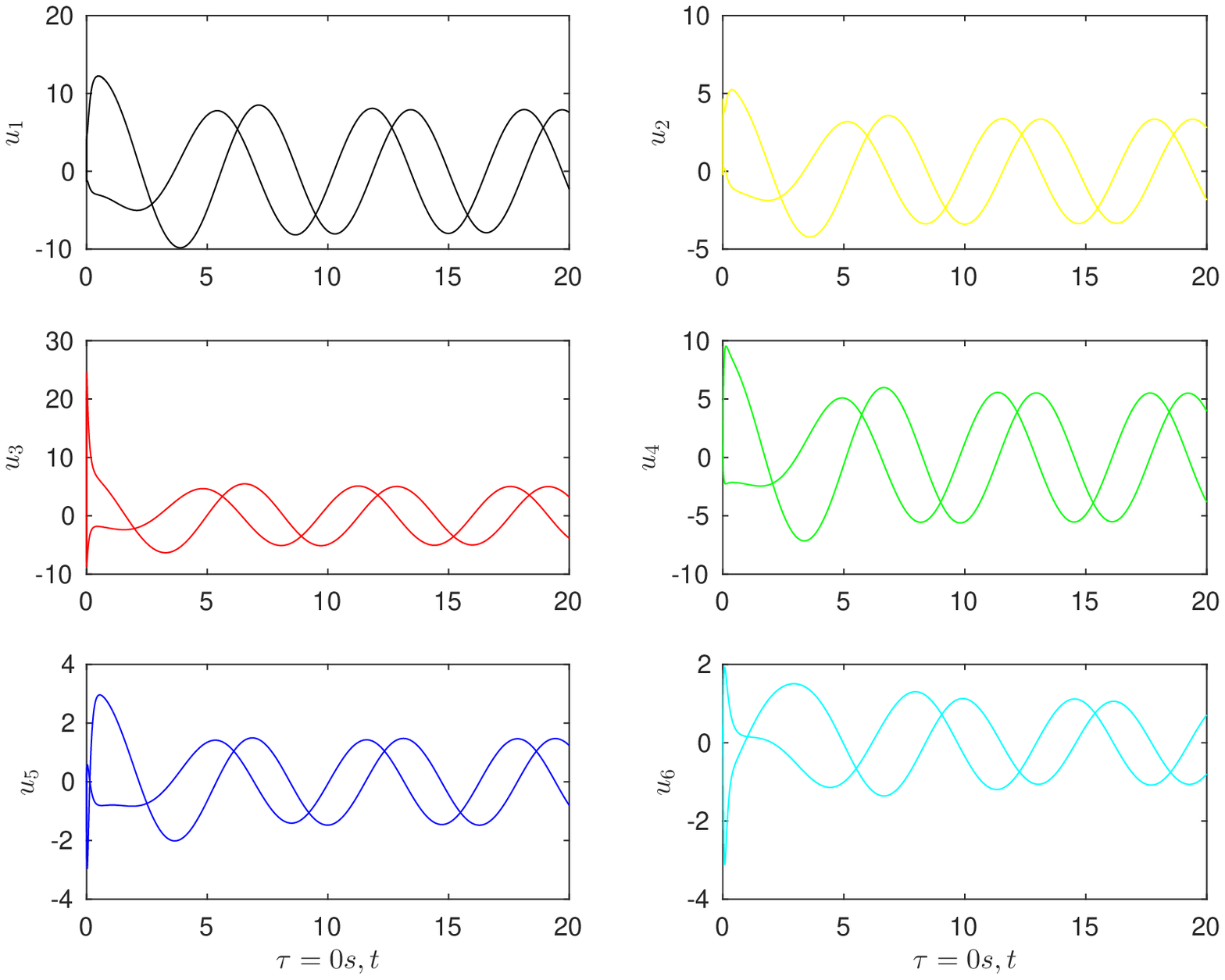}
		\caption{Control input without delay.}
		\label{fig:input_without_delay}
	\end{subfigure} 
	~ 
	\begin{subfigure}[h]{0.47\textwidth}
		\includegraphics[width=\textwidth]{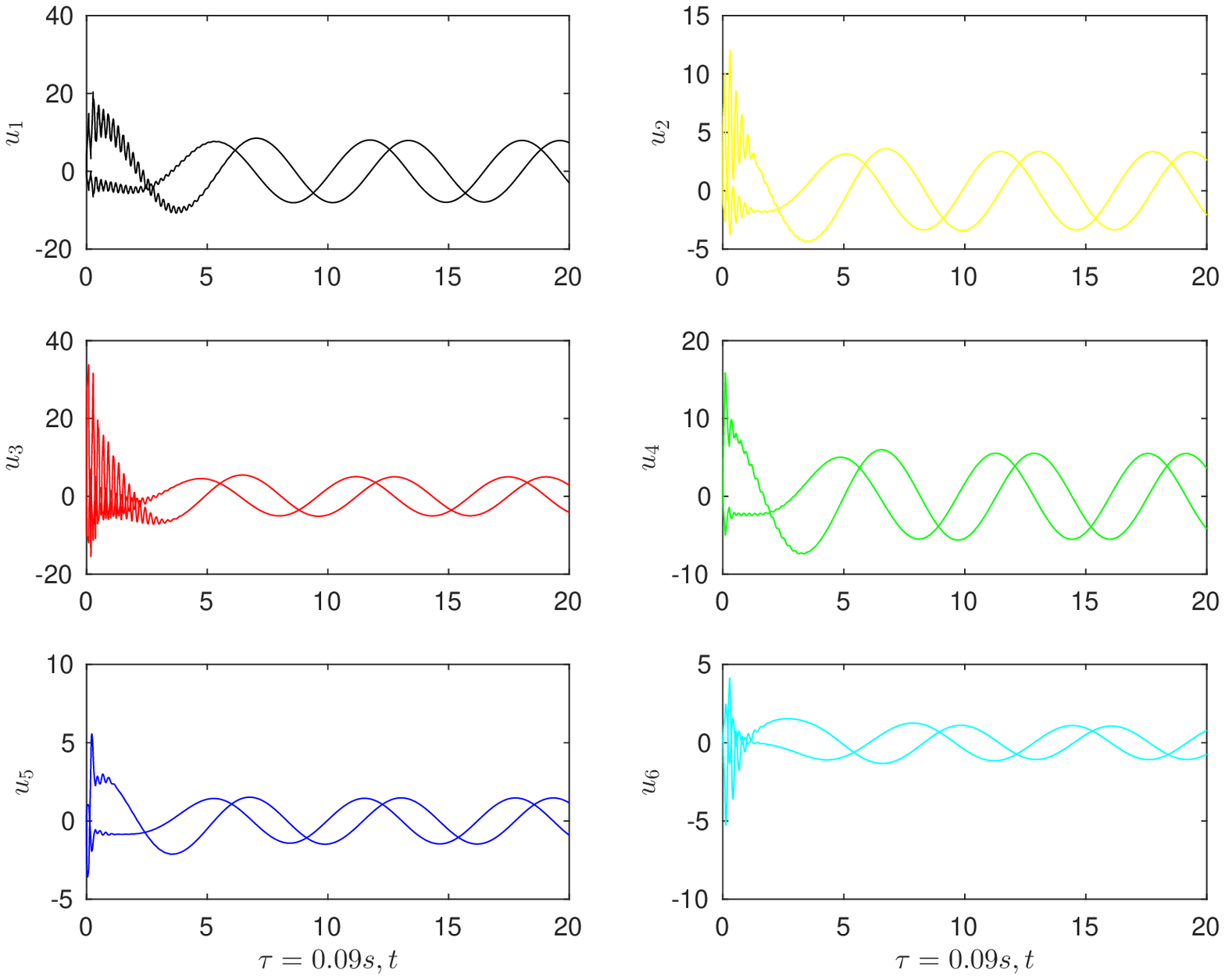}
		\caption{Control input with delay.}
		\label{fig:input_with_delay}
	\end{subfigure}
	~ 
	\begin{subfigure}[h]{0.47\textwidth}
		\includegraphics[width=\textwidth]{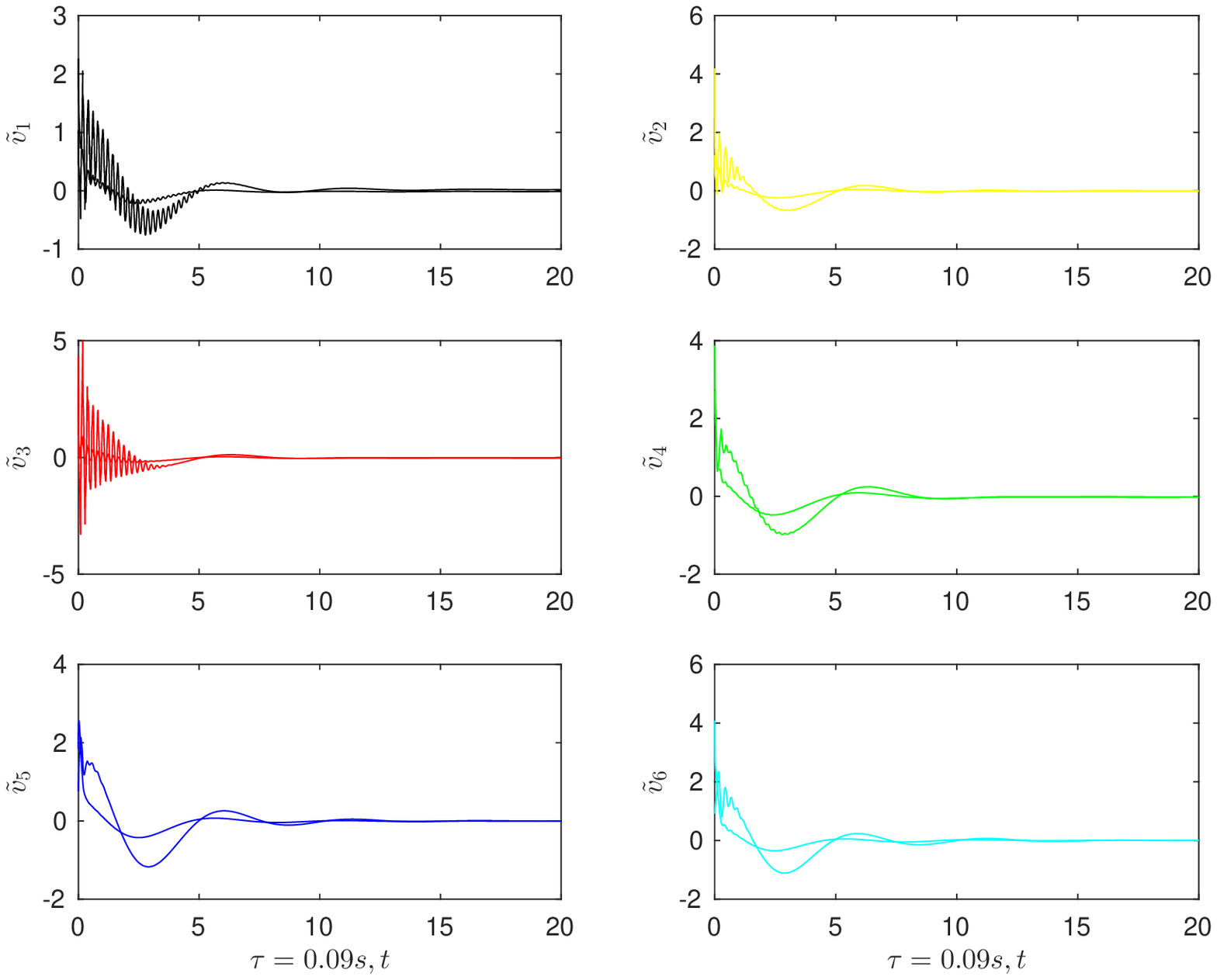}
		\caption{Observer error $ \tilde{v}=v-\tilde{Z} $.}
		\label{fig:observer_error}
	\end{subfigure}
	~ 
	\begin{subfigure}[h]{0.47\textwidth}
		\includegraphics[width=\textwidth]{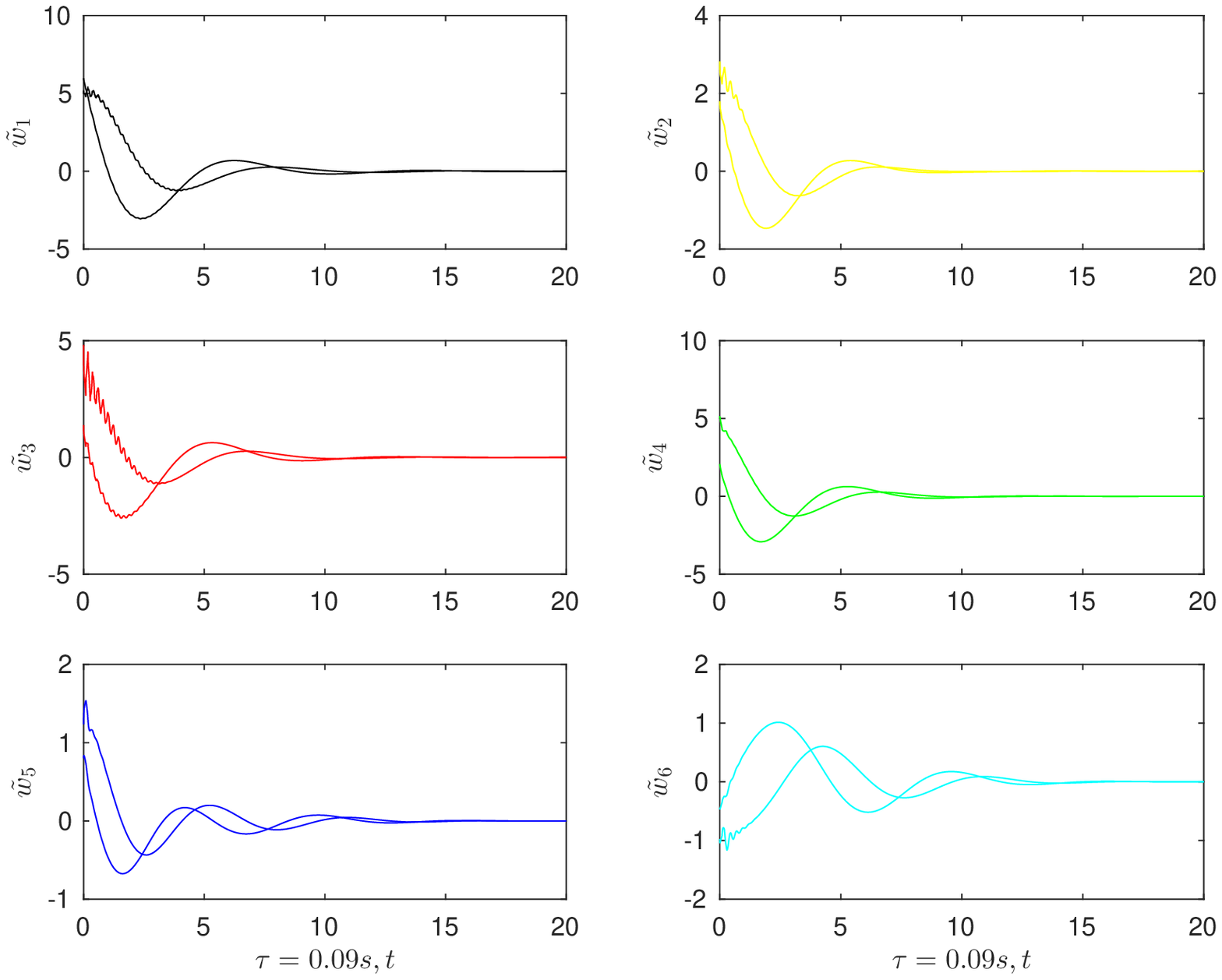}
		\caption{Disturbance observer error $ \tilde{w}=\hat{w}-\bar{w} $.}
		\label{fig:disturbance_observer_error}
	\end{subfigure}
	\caption{Comparison of delay-free and delayed results verifying Theorem~\ref{theorem_disturbance}.}\label{fig:comparison}
\end{figure*}

\begin{figure}[!tb]
	\centering
	\includegraphics[width=0.85\hsize]{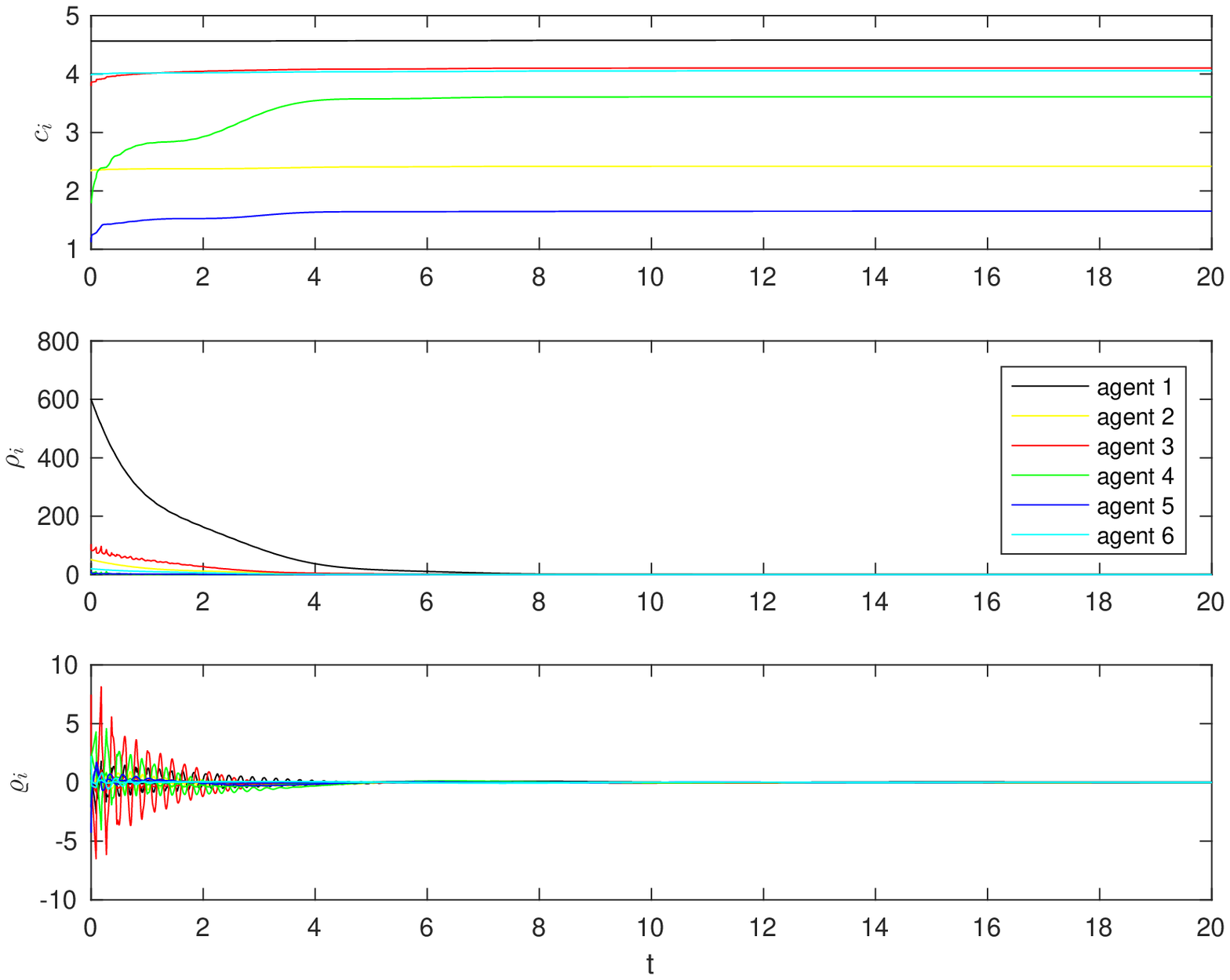}
	\caption{Controller parameters $ c_{i} $ (top), $ \rho_{i} $ (center), $ \varrho_{i} $ (bottom).}
	\label{fig:ci_varrho_rho}
\end{figure}

\textbf{Example 2.}
This example verifies Theorem~\ref{theorem_disturbance_u0}.
Define the leader's bounded input as $ u_{0}(t)=[e^{-t}+1, 2+sin(\frac{t}{2})]^{T} $ and $ \alpha =4, \beta_{1}=1, \epsilon_{i}=0.1, \sigma_{i}=0.005, i\in \mathbb{F} $.
Other initial conditions are the same as the Example 1. Choose $ \mu=2, P>0,Q>I $ and Solve LMIs~\eqref{lmi} and \eqref{lmi_Q}, then
\begin{equation*}
P=\begin{bmatrix}
0.2220 &   0.1066&   -0.1335 &  -0.1098\\
0.1066  &  0.5897  & -0.6273 &  -0.1022\\
-0.1335 &  -0.6273 &   1.2235 &   0.0287\\
-0.1098 &  -0.1022 &   0.0287 &   0.1399
\end{bmatrix},
\end{equation*} 
\begin{equation*}
Q=\begin{bmatrix}
4.0340 &  -0.0000\\
-0.0000  &  2.4367
\end{bmatrix}.
\end{equation*} 

From Fig.~\ref{fig:tiled_x_uub} we can see that the consensus tracking error is indeed uniformly ultimately bounded. We can also tune the controller parameters based on Remark~\ref{remark_tuning_parameters} to control the error as small as possible. 
Fig.~\ref{fig:ci_varrho_rho_with_leader_input} still verifies $ \lim_{t \to \infty} c_{i}(t-\tau)=c_{i}(t), \lim_{t \to \infty} \rho_{i}(t-\tau)=\rho_{i}(t)$ and $ \lim_{t \to \infty} \varrho_{i}(t-\tau)=\varrho_{i}(t), i\in \mathbb{F} $ with time goes on.
In addition, the trajectories of leader and followers are illustrated in Fig.~\ref{fig:state_trajectory}.

\begin{figure}
	\centering
	\begin{subfigure}[h]{0.47\textwidth}
		\includegraphics[width=\textwidth]{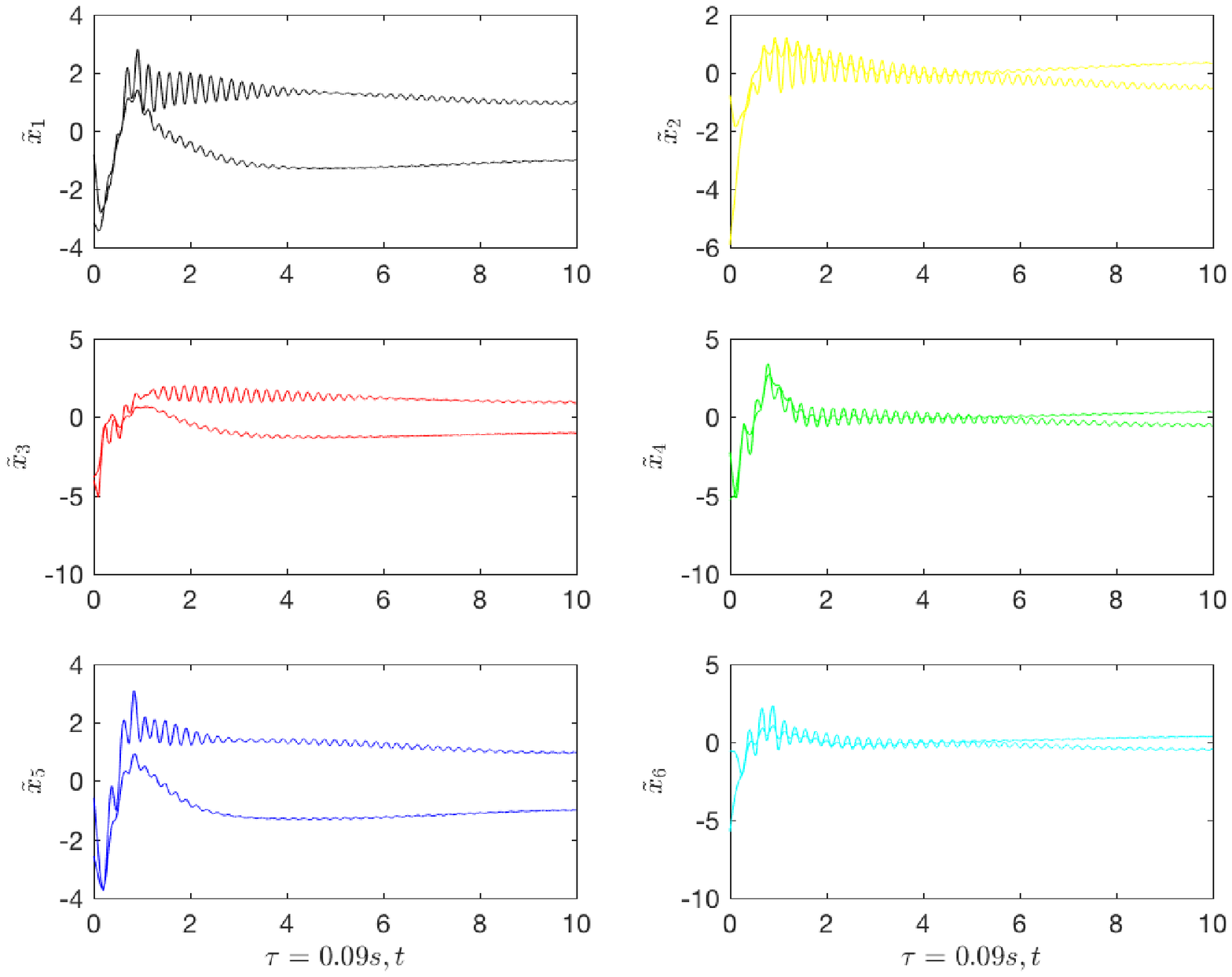}
		\caption{Uniformly ultimately bounded error.}
		\label{fig:tiled_x_uub}
	\end{subfigure}
	~ 
	\begin{subfigure}[h]{0.47\textwidth}
		\includegraphics[width=\textwidth]{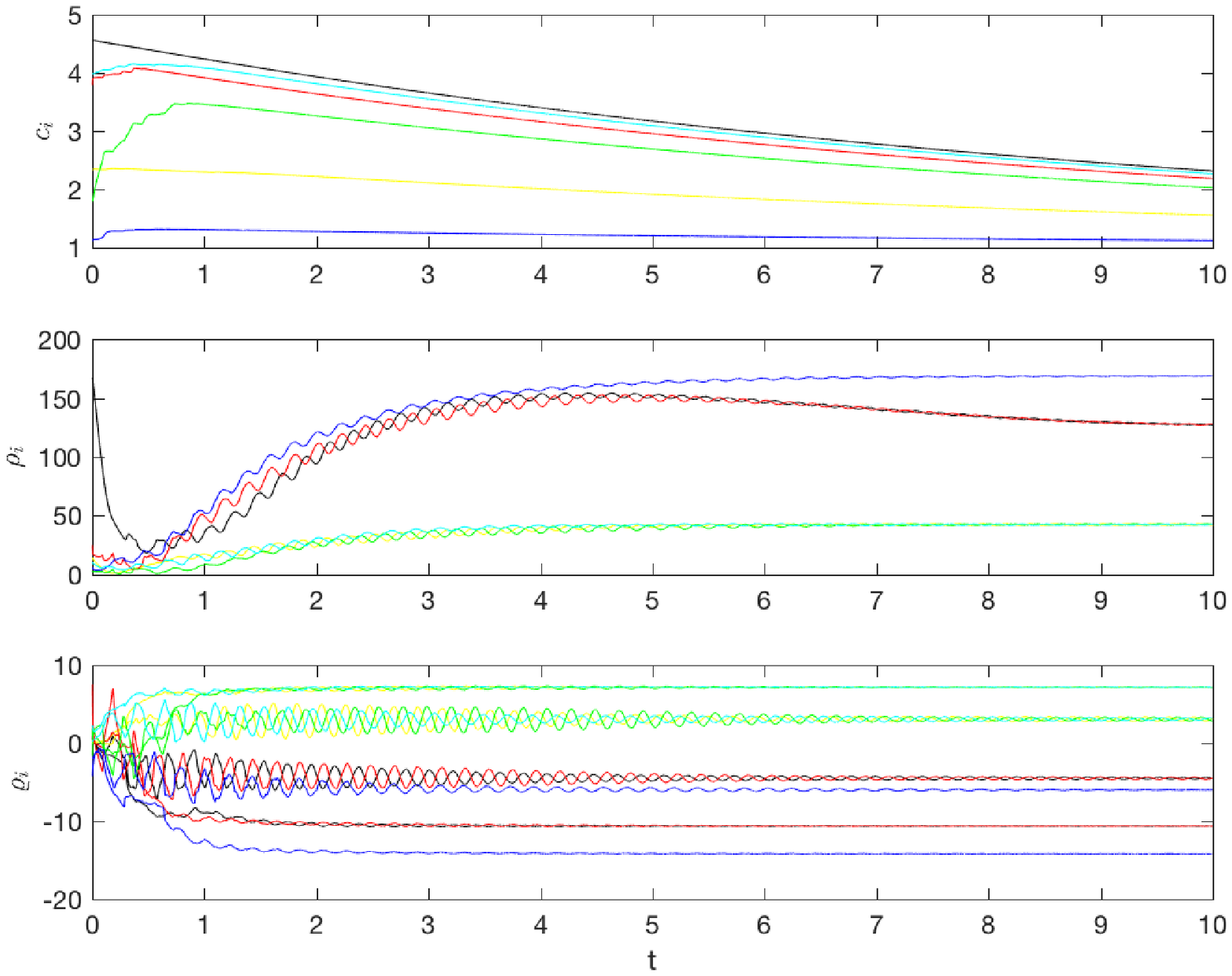}
		\caption{$ c_{i} $ (top), $ \rho_{i} $ (center), $ \varrho_{i} $ (bottom).}
		\label{fig:ci_varrho_rho_with_leader_input}
	\end{subfigure}
	~
	\begin{subfigure}[h]{1\textwidth}
		\includegraphics[width=\textwidth]{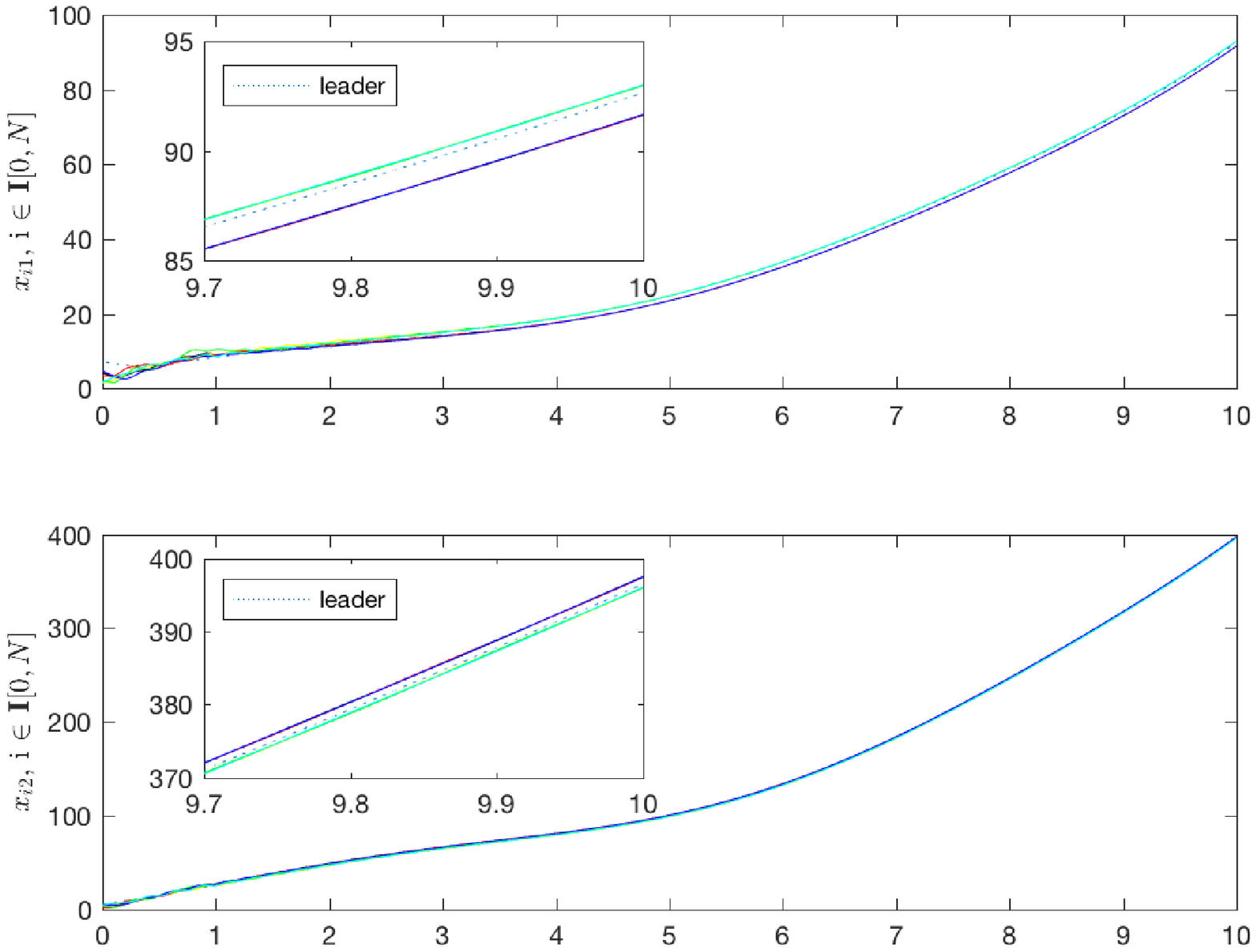}
		\caption{State trajectories.}
		\label{fig:state_trajectory}
	\end{subfigure}
	\caption{Consensus tracking with leader's bounded input verifying Theorem~\ref{theorem_disturbance_u0}.}\label{fig:with_leader_input}
\end{figure}

\section{Conclusion}\label{conclusion}
Designing the fully distributed consensus controller for MASs with an unknown leader subject to input delay and disturbances under the directed communication topology is challenging and important. To complete such a task, novel adaptive predictive extended state observers (ESOs) are proposed using the relative state signals of neighbors. 
The detail steps about how to design the variables for nonlinear function $ z(\cdot) $ in~\eqref{z_discontinuous} is presented.
Considering the various heterogeneity in reality, future work will focus on heterogeneous linear MAS consensus tracking with unknown leader, disturbances and time-varying delay without knowing its upper bound.

\section{Acknowledgements}\label{acknowledgements}
This work was supported by China Scholarship Council, the National Natural Science Foundation of China under
Grants 61503016 and 61403019, the Fundamental Research
Funds for the Central Universities under Grants  2017JBM067, and YWF-15-SYS-JTXY-007,YWF-16-
BJ-Y-21, the National Key R\&D  Program of China under Grant 2017YFB0103202.

\bibliographystyle{model1-num-names}
\bibliography{refrence}







\end{document}